\documentclass[journal,11pt,onecolumn,draftclsnofoot,]{IEEEtran}
\ifCLASSINFOpdf
\else
\fi
\usepackage{mathtools}
\usepackage{blindtext}
\usepackage{enumitem}
\usepackage{bbm, dsfont}
\usepackage{color}
\usepackage{xcolor}
\usepackage{tcolorbox}
\usepackage{balance}
\usepackage{amsthm}
\usepackage{nicematrix,tikz}
\usepackage[utf8]{inputenc}
\usepackage{booktabs,caption}
\usepackage[flushleft]{threeparttable}
\usepackage[english]{babel}
\usepackage{amssymb}
\usepackage[utf8]{inputenc} 
\usepackage[T1]{fontenc}
\usepackage{dblfloatfix}
\usepackage{url}
\usepackage{ifthen}
\usepackage{cite}
\usepackage{graphicx}
\usepackage{psfrag}
\usepackage{blkarray}
\DeclarePairedDelimiter\ceil{\lceil}{\rceil}

\def\nn{\nonumber}

\usepackage{accents}
\newcommand{\dbtilde}[1]{\accentset{\approx}{#1}}

\allowdisplaybreaks
\begin{document}
\title{
Finite-State Channels with Feedback and State Known at the Encoder
\thanks{
E. Shemuel was supported by the Ministry of Science and Technology of Israel. This work was supported by the German Research Foundation (DFG) via the German-Israeli Project Cooperation [DIP] and by the ISF research grant 818/17. 
The material in this paper was presented in part at the
56th Annual Allerton Conference on Communication, Control, and Computing, Monticello, IL, USA, October 2018, and at the IEEE International Symposium on Information Theory, Los Angeles, CA, USA, June 2020.
E. Shemuel and H. H. Permuter are with the Department of Electrical and
Computer Engineering, Ben-Gurion University of the Negev, Beer-Sheva
8410501, Israel (e-mail: els@post.bgu.ac.il; haimp@bgu.ac.il). O. Sabag is with the School of Engineering and Computer Science, 
The Hebrew University of Jerusalem, Jerusalem, Israel (e-mail: oron.sabag@mail.huji.ac.il).}}
\author{Eli Shemuel, Oron Sabag, Haim H. Permuter}
\maketitle
\begin{abstract}
We consider finite state channels (FSCs) with feedback and state information known causally at the encoder. 
This setting is quite general and includes: a memoryless channel with i.i.d. state (the Shannon strategy), Markovian states that include look-ahead (LA) access to the state and energy harvesting. 
We characterize the feedback capacity of the general setting as the directed information between auxiliary random variables with memory to the channel outputs. 
We also propose two methods for computing the feedback capacity: (i) formulating an infinite-horizon average-reward dynamic program; and (ii) a single-letter lower bound based on auxiliary directed graphs called $Q$-graphs. We demonstrate our computation methods on several examples. In the first example, we introduce a channel with LA and derive a closed-form, analytic lower bound on its feedback capacity. Furthermore, we show that the mentioned methods achieve the feedback capacity of known unifilar FSCs such as the trapdoor channel, the Ising channel and the input-constrained erasure channel. Finally, we analyze the feedback capacity of a channel whose state is stochastically dependent on the input.
\end{abstract}

\begin{IEEEkeywords}
channel capacity, channels with feedback, dynamic programming, finite-state channel, Q-graphs.  
\end{IEEEkeywords}
% \fi
\IEEEpeerreviewmaketitle

\newtheorem{question}{Question}
\newtheorem{claim}{Claim}
\newtheorem{guess}{Conjecture}
\newtheorem{definition}{Definition}
\newtheorem{fact}{Fact}
\newtheorem{assumption}{Assumption}
\newtheorem{theorem}{Theorem}
\newtheorem{lemma}{Lemma}
\newtheorem{ctheorem}{Corrected Theorem}
\newtheorem{corollary}{Corollary}
\newtheorem{proposition}{Proposition}
\newtheorem{remark}{Remark}
\newtheorem{example}{Example}

\def\cS{{\mathcal S}}
\def\cX{{\mathcal X}}
\def\cU{{\mathcal U}}
\def\cV{{\mathcal V}}
\def\cQ{{\mathcal Q}}
\def\cY{{\mathcal Y}}

\section{Introduction}

\IEEEPARstart{T}{he} capacity of discrete memoryless channels (DMCs) with an independent and identically distributed (i.i.d.) state, where causal state information (SI) is known at the encoder, was studied by Shannon \cite{Shannon58}. Furthermore, Shannon showed that feedback does not increase the capacity of a DMC \cite{shannon56}, which holds also when causal SI is available. However, this is not the case for channels with memory. Channels with memory, which are common in wireless communication \cite{Sadeghi,ZHANG,turin1990performance,FSCTransWirelessComm,FSCTransWirelessComm1,pimentel2004finite,zhong2008model}, 
molecular communication \cite{MolecFSCTransComm,MolecularSurvey}
and magnetic recordings \cite{Immink}, can be often described by the finite-state channel (FSC) model \cite{Blackwell58,Gallager68,GBAA,Ziv85,PfisterISI,RGray}. 
The memory in FSCs is encapsulated in a channel state with a finite set of values.
%Multi-letter expressions for the feedback capacity of FSCs have been studied in \cite{chenberger,tatikonda2008capacity,PermuterWeissmanGoldsmith09,dabora2012capacity}. %However, no work has been done on FSCs with stochastic channel state evolution when the state is known at the encoder but not at the decoder. 
In this paper, we generalize Shannon's work to the case of FSCs, i.e., we study FSCs with feedback and causal SI known at the encoder, as depicted in Fig. \ref{fig:setting}. 

This setting we study covers many interesting scenarios. One scenario is a channel with i.i.d. state when the SI is available at the encoder in advance with some finite look-ahead (LA) \cite{weissman2006source,DasNarayan,latticeStrategies}. The question whether feedback increases the capacity of this scenario is an open problem. Our setting covers this scenario with feedback since an i.i.d. state with a finite LA can be viewed as a Markovian state process causally known at the encoder. 
Additionally, our setting covers scenarios in which the state is input-dependent, that is, the state evolution depends on the channel inputs. A well-known problem in which the state is input-dependent is the energy-harvesting (EH) model \cite{ozel2012optimal,tutuncuoglu2012optimum,mao2013capacity,dong2014approximate,jog2014energy,shaviv2016capacity,tutuncuoglu2017binary}, motivated by many emerging wireless networks. 
%%%%%%%%%%%%%%%%%%%%%%%%%%%
%Generally speaking, this problem describes a communication channel where the encoder is powered by a finite battery being charged by an exogenous energy source according to an energy arrival process. The transmitted symbols are constrained by the stored energy remaining in the battery. The exact capacity of the EH model has remained an open problem for any finite, non-zero size of the battery, even when the channel is noiseless. However, multi-letter expressions for the capacity of the EH model when a DMC is assumed were given in \cite{mao2013capacity}, and in \cite{tutuncuoglu2017binary} where the channel is noiseless.
%%%%%%%%%%%%%%%%%%%%%%%%%%%
%References \cite{dong2014approximate,jog2014energy,shaviv2016capacity} consider an additive white Gaussian noise (AWGN) channel and obtain bounds on the capacity. References \cite{mao2013capacity,tutuncuoglu2013binary,tutuncuoglu2014improved,tutuncuoglu2017binary} treat a DMC.
%Mao and Hassibi \cite{mao2013capacity} derived a multi-letter expression for the capacity that is hard to compute. Such a DMC with an input constraint can be modeled as an equivalent FSC without an input constraint. Additionally, viewing the energy remained in the battery as a channel state known causally at the encoder makes the EH model with feedback a special case of our setting.
The EH model can be viewed as a FSC where the
channel state is the current battery level governed by the channel inputs and the charging process,
%remaining energy in the battery plays the role of a channel state known causally at the encoder,
thus it is covered by our setting in the presence of feedback. One more input-dependent scenario covered by Fig. \ref{fig:setting} is Noisy Output is the STate (NOST) channels, i.e., channels where the state is stochastically dependent on the channel output, with feedback and causal SI known at the encoder. The capacity of this scenario was derived in \cite{NOST} as a single-letter expression.

\begin{figure}[t]
\begin{center}
\begin{psfrags}
    \psfragscanon
    \psfrag{E}[][][1]{$M$}
    \psfrag{S}[][][1]{\begin{tabular}{@{}l@{}}
    $S^{i-1}$
    \end{tabular}}
    %{\\}
    \psfrag{A}[\hspace{2cm}][][1]{Encoder}
	 \psfrag{F}[\hspace{1cm}][][1]{$X_i$}
	 \psfrag{B}[\hspace{2cm}]{{$P(s_i,y_i|x_i,s_{i-1})$}}
	 \psfrag{G}[][][1]{$Y_i$}
	 \psfrag{C}[\hspace{2cm}][][1]{Decoder}
	 \psfrag{K}[][][1]{$\hat{M}$}
	 \psfrag{H}[\hspace{2cm}][][1]{$Y_i$}
	 \psfrag{D}[\hspace{2cm}][][1]{Unit-Delay}
	 \psfrag{J}[\vspace{2cm}\hspace{2cm}][][1]{$Y_{i-1}$}
	 \psfrag{L}[\hspace{2cm}][][1]{Finite-State Channel}
	 \psfrag{I}[][][1]{}
	 %\psfrag{tag}[][][<scale>]{Latex Text}
\includegraphics[scale=0.8]{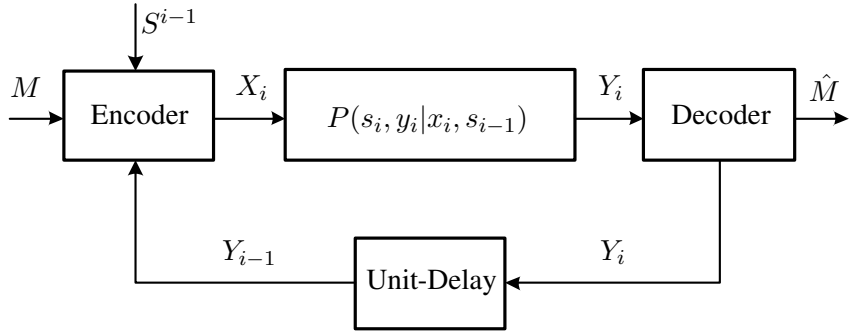}
\caption{
FSCs with feedback and SI known at the encoder. Note that the channel output at time $i$ depends on $S_{i-1}$. }
\label{fig:setting}
\psfragscanoff
\end{psfrags}
\end{center}
\end{figure}
%Throughout the paper, we examine three examples which are special problems of the general setting using either one of the two tools. The first example is the open problem of determining the capacity of the energy-harvesting (EH) channel with a finite battery \cite{ozel2012optimal,tutuncuoglu2012optimum,mao2013capacity,dong2014approximate,jog2014energy,shaviv2016capacity,tutuncuoglu2017binary}, motivated by many future wireless networks.

We derive a multi-letter capacity formula for the general setting
expressed as the directed information between a sequence of auxiliary random variables (RVs) and the sequence of the channel outputs. As in the case of i.i.d. states, we show that the channel input is a function of the auxiliary RV and the state. However, in our setting, the auxiliary RVs should have memory. 
%that consists of , that map channel states to channel inputs, and are commonly referred to as \textcolor{blue}{\textit{strategies}}\cite{keshet2008channel,li2016outage,choudhuri2013causal,mirmohseni2009compress,tutuncuoglu2017binary}.
% . For the achievability we revisit a technique from Shannon's work \cite{shannon1958channels} that was later named after him, the \textit{Shannon strategy} . Loosely speaking, the encoder chooses an auxiliary RV that represents a mapping from the state to the channel input. This technique induces a new channel without channel state availability at the encoder. While in Shannon's work the induced channel was a new DMC whose capacity was well-known, applying this technique in our work induces a new FSC with no state availability at either party and with feedback, a setting whose achievability was derived in \cite{PermuterWeissmanGoldsmith09} as a multi-letter expression.
%Due to its usefulness, this coding technique has frequently been exploited in modern publications, e.g., \cite{keshet2008channel,li2016outage,choudhuri2013causal,mirmohseni2009compress,tutuncuoglu2017binary}.
%In our setting, the auxiliary RVs may be in an infinite set, and hence one can compute lower bounds on the capacity by assuming that this set is finite and fixed over time.
Although our capacity expression is multi-letter, we use it to propose sequences of lower and upper bounds on the feedback capacity, whose elements are finite-letter expressions. Furthermore, we develop two methods for computing the feedback capacity. The first method is a formulation of the capacity expression as an infinite-horizon dynamic programming (DP) optimization problem. %we present an infinite-horizon average-reward Markov decision process (MDP) formulating a lower bound on the capacity.
%The first method comprises formulating a lower bound on the capacity as an infinite-horizon dynamic program (DP).
In the second method, we derive a single-letter lower bound based on $Q$-graphs \cite{Sabag_UB_IT}. 
%Both methods are shown to be successful in achieving the capacity of various FSC problems that are special cases of the general setting.

The DP framework was introduced in \cite{Tatikonda2000ControlUC} as a tool for computing the feedback capacity of FSCs. This paved the way for more works that formulated the feedback capacity of certain FSCs as a DP: Markov channels \cite{TatikondaMitter_IT09}, ISI channels \cite{yang2005feedback} 
%**Such channels in which the state is dependent of the input are called channels with intersymbol interference (ISI).
%The work in [9] also demonstrates simplifications of the capacity computation for FSCs in which the channel state can be computed from the channel inputs and outputs.
%while \cite{chen2005capacity} considered FSCs where the state was a deterministic function of the output. %Permuter \textit{et. al}
and the most general class of unifilar FSCs \cite{Permuter06_trapdoor_submit}. For the latter case of unifilar FSCs, the new state is a function of the current state, input and output, and consequently the state is known to the encoder. Thus, even the general class of unifilar FSCs is captured in our framework of Fig. \ref{fig:setting}.
%(the initial state was assumed to be known at the encoder). Consequently, unifilar FSCs include the models in
%\cite{yang2005feedback,chen2005capacity}.
%We note that for unifilar FSCs in which the initial state is known at the encoder, the state can be computed causally at the encoder at any time, thus they can be considered as another special case of our setting. 
%This work provides such a desired framework for any finite, fixed cardinality of the auxiliary RV.
%For any special problem of the general setting, one can implement this framework via the value-iteration algorithm on a machine.
%computer-based simulations of DP provide crucial insights into the feedback capacity of a particular channel. Specifically, implementation of the value-iteration algorithm together with DP simulation can be used to derive an analytic expression for the feedback capacity. The numerical results naturally provide a lower bound on the feedback
%In \cite{tatikonda2008capacity} he capacity can be formulated as a dynamic programming (DP) optimization problem; this has benefits such as efficient algorithms for estimating the capacity and analytical tools for calculating capacity.
%The relationship between the feedback capacity of FSCs and DP first appeared in Tatikonda’s thesis [8]. The need for this formulation arises from difficulties in the computability of the capacity expression as can be seen in (1).

The $Q$-graph technique, introduced in \cite{Sabag_UB_IT}, is another tool for computing lower and upper bounds on the feedback capacity. It maps the receiver’s output sequences to a sequential quantization in a finite set of graphs represented by a directed graph, called a $Q$-graph. In \cite{Sabag_UB_IT}, single-letter lower and upper bounds on the feedback capacity were derived for unifilar FSCs for any $Q$-graph, while in \cite{Graph-Based}
it was shown how to compute them. We implement their method on our multi-letter capacity expression of the general setting in order to derive a single-letter lower bound on the feedback capacity for any $Q$-graph.
%and fixed, finite cardinality of the set of the auxiliary RVs. 

Generally, both the DP and the $Q$-graph tools can be used to compute achievable rates, but for special cases they capture the precise feedback capacity. We show the tightness of the $Q$-graph bound for known unifilar FSCs such as the trapdoor channel \cite{Permuter06_trapdoor_submit}, the Ising channel \cite {elishco2014capacity} and the binary erasure channel (BEC) with inputs constraint \cite{Sabag_BEC}. We also investigate the feedback capacity of a channel with LA SI known at the encoder, and the feedback capacity of a generalized Ising channel whose state is stochastically dependent on the input.
The remainder of the paper is organized as follows.
Section~\ref{sec:problem_def} defines the notation used in this paper and the setting, and provides preliminaries on DP and $Q$-graphs. Section~\ref{sec:main_results} presents the main results. Section~\ref{sec:dpQgraph} focuses on the DP and $Q$-graph methods. In Section~\ref{sec:examples}, we provide several FSCs and study their feedback capacity. Section~\ref{sec:proof_cap} proves our main result.
%of unifilar FSCs are studied and the capacity of the probabilistic symmetric POST channel is derived.
%. Section~\ref{sec:qgraph} focuses on the $Q$-graph lower bound on the capacity.
 %by using either the DP or the $Q$-graph tools.
Finally, Section~\ref{sec:conclusions} concludes this work.
\section{The Communication Setup and Preliminaries}
\label{sec:problem_def}
In this section, we introduce the notation and the communication setup. We then provide preliminaries on DP and the $Q$-graphs. 
\subsection{Notation}
Lowercase letters denote sample values (e.g. $x,y$), and uppercase letters denote discrete RVs (e.g. $X,Y$). Subscripts and superscripts denote vectors in the following way: $x_i^j=(x_i,x_{i+1},...,x_j)$ and $X_i^j=(X_i,X_{i+1},...,X_j)$ for $1\leq i \leq j$. $x^n$ and $X^n$ are shorthand for $x_1^n$ and $X_1^n$, respectively.
We use calligraphic letters (e.g. $\cX,\cY$) to denote alphabets, and $|\cX|$ to denote the cardinality of the alphabet. For two RVs $X,Y$ the probability mass function (PMF) of $X$ is denoted by $P(X=x)$, the conditional PMF of $X=x$ given $Y=y$ is denoted by $P(X=x|Y=y)$, and the joint PMF is denoted by $P(X=x,Y=y)$; the shorthand $P(x),P(x|y),P(x,y)$ are used for the above, respectively. The indicator function is denoted by $\mathbbm{1}(\cdot)$. We use $\oplus$ to denote the binary XOR operation. We define $\bar{a}=1-a$ for some $a \in [0,1]$. For a pair of integers $n\le m$, we define the discrete interval $[n:m]\triangleq \{n,n+1,\dots,m\}$. We use logarithms to base $2$; thus the entropy is measured in bits. 

The \textit{directed information} (DI) between $X^N$ to $Y^N$ conditioned on $S$, introduced by Massey \cite{massey1990causality} and employed with conditioning in \cite{PermuterWeissmanGoldsmith09}, is defined as 
\begin{equation}
I(X^N\rightarrow Y^N|S)\triangleq \sum_{i=1}^{N} I(X^i;Y_i|Y^{i-1},S).
\end{equation}
The \textit{causally conditional distribution} (CCD) conditioned on $s_0$, introduced in \cite{Kramer03,permuter2006capacity} and employed with conditioning in \cite{Permuter06_trapdoor_submit}, is defined as
\begin{equation}
    P(x^N||y^{N-1},s)\triangleq \prod_{i=1}^N P(x_i|x^{i-1},y^{i-1},s). 
\end{equation}
\subsection{The Setting}
We consider FSCs as shown in Fig. \ref{fig:setting}. A FSC consists of finite input, output and channel state alphabets $\cX,\cY,\cS$, respectively.
It is defined by ($\cX \times \cS$, $P_{Y,S^+|X,S}$, $\cY \times \cS$) where $S,S^+$ are the channel state at the beginning and at the end of the transmission, respectively. The initial state is distributed according to $P(s_0)$, and it is available to the encoder but not to the decoder. At time $i$, the encoder has access to the message $m\in \mathcal{M}$, the output feedback and the channel state. The channel is time invariant, and at each time $i$ it has the Markov property
\begin{equation}
P(y_i,s_i|x^i,s_0^{i-1},y^{i-1.},m) = P_{Y,S^+|X,S}(y_i,s_i|x_i,s_{i-1}).
\label{eq:BasicFSCMarkov}
\end{equation}
The encoder's mapping at time $i$ is denoted as 
\begin{align}
\label{eq:enc}
    f_i:\mathcal{M} \times \cS_{0}^{i-1} \times \cY^{i-1} \to \cX,
\end{align}
and the decoder's mapping is
\begin{align}
\label{eq:dec}
    \hat{m}: \cY^n \to \mathcal{M}.
\end{align}
An $(2^{nR},n)$ code is a pair of encoding and decoding mappings \eqref{eq:enc}-\eqref{eq:dec} with a message set $\mathcal{M}=[1:\ceil{2^{nR}}]$, and $M$ is uniformly distributed over $\mathcal{M}$.
%the channel state sequence, causally, $s^{i-1}$ and the feedback samples $y^{i-1}$. The encoder output is denoted by $x_i$ and it is a function of the tuple known to it, i.e., $x_i(m,s^{i-1},y^{i-1})$. The channel output $y_i$ enters the receiver (decoder) and is also fed back with a unit delay to the encoder.
%We use standard definitions of average probability of error, achievable rate and capacity, e.g. \cite{el2011network,cover2012elements}.
A rate $R$ is \textit{achievable} if there exists a sequence of codes $(n,\ceil{2^{nR}})$ such that the \textit{average probability of error} defined as $P_e^{(n)}\triangleq \Pr (\hat{m}\neq m| \text{message $m$ was sent})$ tends to zero as $n\to \infty$. The \textit{capacity} of the setting is defined as the supremum over all achievable rates, and is denoted by $C_{\text{fb-csi}}$. 
%Throughout this paper, we assume that the initial state $s_0$ is available both to the encoder and to the decoder. 
Furthermore, we assume that the FSC is \textit{strongly connected}.
\begin{definition}[Connectivity] \cite[Def.~2]{Permuter06_trapdoor_submit}
A FSC is strongly connected if for all $s',s \in \cS$ there exist $T(s)$ and input distribution of the form $\{P(x_i|s_{i-1})\}_{i=1}^{T(s)}$ that may depend on $s$, such that $\sum_{i=1}^{T(s)} P(S_i=s|S_0=s')>0$.
\end{definition}
% A FSC is strongly connected if for any state $s\in \cS$ and $\epsilon > 0$ there exist an integer $T(s,\epsilon)$ and an input distribution $\{P(x_i|s_{i-1})\}_{i=1}^{T(s,\epsilon)}$ such that 
% \begin{equation}
%     \Pr (S_{T(s,\epsilon)}=s|S_0=s')>1-\epsilon
% \end{equation}
% for any initial state $s'$.

\subsection{Average-Reward Dynamic Programming}
\label{subsection:prelDP}
%An MDP formulation enables one to numerically compute and analytically derive achievable rates. This is very useful while trying to compute multi-letter expressions such as the directed information.
A DP is defined by a septuple $(\mathcal{Z},\mathcal{U},\mathcal{W},F,P_Z,P_w,g)$.
We consider a discrete-time dynamic system evolving according to
\begin{equation*}%\label{eq_DP}
  z_i=F(z_{i-1},a_i,w_i),\  i=1,2,\dots
\end{equation*}
Each state, $z_i$, takes values in a Borel space $\mathcal{Z}$.
Each action, $a_i$, takes values in a compact subset $\mathcal{A}$ of a Borel space.
Each disturbance, $w_i$, takes values in a measurable space $\mathcal{W}$, and is drawn from a distribution $P_w(\cdot|z_{i-1},a_i)$ that depends on the state $z_{i-1}$, and action $a_i$.
The initial state, $z_0$, is drawn from a distribution $P_Z$. All functions considered in this section are assumed to be measurable. The history, \textcolor{black}{$h_{i}=(z_0,w_1^{i-1})$}, summarizes information available to the controller at time \textcolor{black}{$i$, prior to the selection of the $i$th action}. At time \textcolor{black}{$i$}, the controller selects the action, $a_i$, by a function \textcolor{black}{$\mu_{i}$} that maps histories to actions, i.e., \textcolor{black}{$a_i = \mu_{i}(h_{i})$}. Given a policy, denoted by $\pi=\{\mu_1,\mu_2,\dots\}$, and the history, \textcolor{black}{$h_{i}$}, one can compute the actions vector, $a^i$, and the past state vector of the system, $z^{i-1}$.

Given a bounded reward function $g: \cal{Z} \times \cal{A} \rightarrow \mathbb{R}$,
the objective is to maximize the infinite-horizon average reward. For a policy $\pi$, it is defined by
\begin{equation}
\label{average_reward}
\rho_{\pi} = \liminf_{N \to \infty} \frac{1}{N} \mathbb{E}_{\pi}\left\{\sum_{i=1}^{N}
g(Z_{i-1}, \mu_{i}(H_{i})) \right\},
\end{equation}
where the subscript $\pi$ indicates that actions $a_i$ are generated
by the policy $\pi$. The optimal average reward is defined by
\begin{equation}
\rho^*= \sup_{\pi} \rho_\pi.
\end{equation}

\subsection{$Q$-graphs}
\label{subsection:prelQgraph}
\begin{figure}[t]
\begin{center}
\begin{psfrags}
    \psfragscanon
    \psfrag{Q}[][][1]{$Y=1$}
    \psfrag{E}[][][1]{$Y=0$}
    \psfrag{F}[][][1]{$Y=?$}
    \psfrag{O}[][][1]{$Y=0/?/1$}
    \psfrag{L}[][][1]{$Q=2$}
    \psfrag{H}[][][1]{$Q=1$}
    \includegraphics[scale = 0.5]{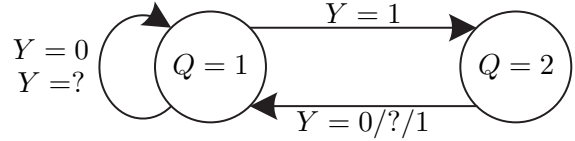}
    \caption{An example of a $Q$-graph with $|\mathcal{Q}|=2$ and $\mathcal{Y}=\{0,1,?\}$. }
    \label{fig:q_graph_ex}
    \psfragscanoff
\end{psfrags}
\end{center}
\end{figure}
The $Q$-graph is a technique that enables us to derive a single-letter lower and upper bounds on the feedback capacity out of a multi-letter expression, based on \textit{Quantized-graphs} ($Q$-graphs). A $Q$-graph is defined as a directed, connected graph with a finite number of nodes. Each node takes a different value $q \in \cQ$, and it has exactly $|\cY|$ outgoing edges that are labeled with distinct symbols from $\cY$. From the $Q$-graph definition, it follows that given an initial node, $q_0$, and an output sequence, $y^i$, walking along the corresponding labeled edges determines uniquely a final node, $q_i$. This induced mapping can be represented by $\Phi_{i}:{\cY}^{i}\to \cQ$ as well as by a time-invariant function, $g:\cQ\times\cY\to\cQ$, where the current graph node and the channel output determine a new node. An example of a $Q$-graph is depicted in Fig.~\ref{fig:q_graph_ex}.
%For any $q,q^+ \in \cQ$ and $y \in Y$, there exists an edge $q \to q^+$ with label $y$ if $q^+=g(q,y)$.

\section{Main Results}
\label{sec:main_results}
In this section, we present our main results. The following two theorems characterize the capacity of the setting.

\begin{theorem}
\label{theorem:capDI}
The feedback capacity of a strongly connected FSC with SI known causally at the encoder is given by
\begin{align}
    C_{\text{fb-csi}}&=\lim_{N \to \infty} \frac{1}{N}  \max_{P(u^N||y^{N-1})} I(U^N \to Y^N)\label{eq:capDI},
\end{align}
% \begin{subequations}
% \begin{alignat}{2}
% C_{\text{fb-csi}}&=\lim_{N \to \infty} \frac{1}{N}  \max_{P(u^N||y^{N-1},s_0)} I(U^N \to Y^N)\label{eq:capDI} \\
% %&= \lim_{N \to \infty} \frac{1}{N}  \max_{P(u^N||y^{N-1})} \max_{s_0} I(U^N \to Y^N|s_0),\label{eq:capDI}
% \end{alignat}
% \end{subequations}
where $\{U_i\}_{i \ge 1}$ are auxiliary RVs with $|\cU|=|\cX|^{|\cS|}$,
% that indicate all possible mappings 
% the set of all functions
% that map all the possible  
%  where $u$ indicates the index of a specific function of all $|\cX|^{|\cS|}$ function.
% that map $X_i=f(U_i,S_{i-1})$,
and the joint distribution is given by 
\begin{align}
    \label{eq:joint_distDI}
    P(x^N,s_0^N,y^N,u^N)&=P(s_0)P(u^N||y^{N-1})\prod_{i=1}^N\mathbbm{1}\{x_i = f(u_i,s_{i-1})\}P_{Y,S^+|X,S}(y_i,s_i|x_i,s_{i-1})
    \text{.}
\end{align}
Each  $u\in \cU$ corresponds to a distinct function from the set $\{f_u (s):\cS \to \cX \}$.
%encapsulated by $x_i=f(u_i,s_{i-1})$
\end{theorem}

The feedback capacity can also be expressed as follows.
\begin{theorem}
\label{theorem:capTrunc}
The feedback capacity of a strongly connected FSC with SI known causally at the encoder is given by
\begin{align}
C_{\text{fb-csi}}&= 
\lim_{N \to \infty} \frac{1}{N}  \max_{\substack{\{P(u_i|u_{i-1},y^{i-1})\}_{i=1}^N\\x_i=f(u_i,s_{i-1})}} \sum_{i=1}^{N} I(U_i,U_{i-1};Y_i|Y^{i-1}),  \label{eq:capTrunc}
\end{align}
% \begin{subequations}
% \begin{alignat}{2}
% C_{\text{fb-csi}}&= 
% \lim_{N \to \infty} \frac{1}{N}  \max_{\substack{\{P(u_i|u_{i-1},y^{i-1})\}_{i=1}^N\\x_i=f(u_i,s_{i-1})}} \min_{s_0} \sum_{i=1}^{N} I(U_i,U_{i-1};Y_i|Y^{i-1},s_0),  \label{eq:capTrunc} \\
% &=\lim_{N \to \infty} \frac{1}{N}  \max_{\substack{\{P(u_i|u_{i-1},y^{i-1})\}_{i=1}^N\\x_i=f(u_i,s_{i-1})}} \max_{s_0} \sum_{i=1}^{N} I(U_i,U_{i-1};Y_i|Y^{i-1},s_0),  \label{eq:capTrunc}
% \end{alignat}
% \end{subequations}
where $\{U_i\}_{i \ge 1}$ are auxiliary RVs, and the joint distribution is given by 
\begin{align}
    \label{eq:joint_distTrunc}
    P(x^N,s_0^N,y^N,u^N)=P(s_0)\prod_{i=1}^N
    P(u_i|u_{i-1},y^{i-1})\mathbbm{1}\{x_i = f(u_i,s_{i-1})\}P_{Y,S^+|X,S}(y_i,s_i|x_i,s_{i-1}) \text{.}
\end{align}
\end{theorem}
%%%%%%%%%%%%
% Since the rates with $\min_{s_0}$ and $\max_{s_0}$ are equal, it is concluded that the feedback capacity can be evaluated with any arbitrary $s'_0\in \cS$ and optimizing $P$ and $f$, e.g., 
% \begin{align}
% C_{\text{fb-csi}}&= 
% \lim_{N \to \infty} \frac{1}{N}  \max_{\substack{\{P(u_i|u_{i-1},y^{i-1})\}_{i=1}^N\\x_i=f(u_i,s_{i-1})}} \sum_{i=1}^{N} I(U_i,U_{i-1};Y_i|Y^{i-1},s_0').
% \end{align}
% Furthermore, $C_{\text{fb-csi}}$ is independent of $P(s_0)$. 
%%%%%%%%%%%%%%%%
%$|\cU_i|=|\cX|^{{|\cS|}^i}$
The capacity expressions in Theorems \ref{theorem:capDI} and \ref{theorem:capTrunc} provide two alternative capacity expressions for FSCs with feedback, and their corresponding joint distributions, \eqref{eq:joint_distDI} and \eqref{eq:joint_distTrunc}, imply that $X_i$ depends on $S_{i-1}$ and not on $S^{i-2}$ via the time-invariant function $f:\cU \times \cS \to \cX$.
% , and also imply the Markov chain %the FSC Markov property \eqref{eq:BasicFSCMarkov} can be rewritten with the auxiliary RVs $U_i$ as
% \begin{align}
% (Y_i,S_i)-(X_i,S_{i-1})-U_i. %,X^{i-1},Y^{i-1},S^{i-2}).
% \end{align}
To compare the objective functions and maximization domains between \eqref{eq:capDI} and \eqref{eq:capTrunc}, recall that the DI and the CCD
in the former can be written as $I(U^N \to Y^N)\triangleq \sum_{i=1}^N I(U^i;Y_i|Y^{i-1})$ and $P(u^N||y^{N-1})\triangleq \prod_{i=1}^N P(u_i|u^{i-1},y^{i-1})$, respectively.
%However, there is a trade-off between the two alternatives. On the one hand, \eqref{eq:capTrunc} is simpler since for all $i$ it has $(U_i,U_{i-1})$ in the sum and the maximization is over $\{P(u_i|u_{i-1},y^{i-1})\}$, while \eqref{eq:capDI} has $U^i$ in the sum of the directed information, and the maximization  is over $\{P(u_i|u^{i-1},y^{i-1})\}$, i.e., it is also conditioned on $U^{i-2}$.
%thus the distinction between them becomes clear. In the first capacity expression, the input 
On the other hand, \eqref{eq:capDI} has a finite cardinality bound for for all $\{U_i\}_{i\ge1}$, while in \eqref{eq:capTrunc} their cardinality may be unbounded. 
We note that in Theorem \ref{theorem:capDI}, there is no maximization over the functions $f$ as the cardinality of $|\cU| = |\cX|^{|\cS|}$ covers all possible mappings from $\cS$ to $\cX$ (called \textit{strategies}).

\begin{lemma}
    \label{cor:TheoremsEq}
Both alternatives, i.e., Theorems \ref{theorem:capDI} and \ref{theorem:capTrunc}, can be shown to be equal.
\end{lemma}
\begin{proof}[Proof sketch]
For distinction, we rename the auxiliary RVs of Theorem \ref{theorem:capTrunc} with $\{V_i\}_{i \ge 1}$, while $\{U_i\}_{i \ge 1}$ are remained to denote the auxiliary RVs of Theorem \ref{theorem:capDI}. By defining $V_i=(V_{i-1},U_i)$, Theorem \ref{theorem:capDI}
becomes Theorem \ref{theorem:capTrunc} with $|\cV_i|= |\cX|^{{|\cS|}^i}$ as for all $i$, $V_i$ recovers all possible strategies $\{f_u (s):\cS \to \cX \}$ until time $i$, that were encapsulated by $x_j=f(u_j,s_{j-1})$ for $j=1,\dots,i$. 
\end{proof}
%\textcolor{blue}{Hence, when considering $|\cV_i|= |\cX|^{{|\cS|}^i}$, there is no maximization of choosing $f(\cdot)$ in \eqref{eq:capTrunc}, since all past strategies are taken}, and the equivalence of both Theorems is concluded.
Due to Lemma \ref{cor:TheoremsEq}, we prove only Theorem \ref{theorem:capDI} in Section~\ref{sec:proof_cap}. Although both Theorems consist of a multi-letter capacity expressions, we utilize them to derive computable lower and upper bounds as given in the following results. First, we obtain sequences of achievable rates and upper bounds on the feedback capacity that are computable for any positive integer $N$, as given in the next theorem. Let $\underline{C}_N$ and $\overline{C}_N$ denote
\begin{align}
    \underline{C}_N &= \frac{1}{N} \max_{P(u^N||y^{N-1})} \min_{s_0} I(U^N \to Y^N|s_0), \\
    \overline{C}_N &= \frac{1}{N} \max_{P(u^N||y^{N-1})} \max_{s_0} I(U^N \to Y^N|s_0),
\end{align}
where $|\cU|=|\cX|^{|\cS|}$, and the joint distribution is 
\begin{align}
    P(x^N,s^N,y^N,u^N|s_0)&=P(u^N||y^{N-1})\prod_{i=1}^N\mathbbm{1}\{x_i = f(u_i,s_{i-1})\}P_{Y,S^+|X,S}(y_i,s_i|x_i,s_{i-1}).
\end{align}
\begin{theorem}
\label{theorem:anyN_LB_UB}
The feedback capacity of any FSC $P_{Y,S^+|X,S}$ with SI known causally at the encoder is bounded by
\begin{align}
\label{eq:corLBUB}
\underline{C}_N-\frac{\log |\cS|}{N}
\le C_{\text{fb-csi}} \le  \overline{C}_N+\frac{\log |\cS|}{N}, \quad N=1,2,\dots
\end{align}

\end{theorem}
The proof of Theorem \ref{theorem:anyN_LB_UB} is given in Section~\ref{sec:proof_cap}. Notice that it holds for any FSC, not necessarily connected.

From Theorem \ref{theorem:capTrunc}, we also derive computable lower bounds using the DP and the $Q$-graph methods that were introduced in
Sections \ref{subsection:prelDP} and \ref{subsection:prelQgraph}, respectively, as shown in the following theorems.
%. For any finite cardinality $|\cU|$, \eqref{eq:capTrunc} can be considered as a lower bound on the feedback capacity. The methods are provided in the next theorems.
%%%%% ORON: remove.
%in which the limit in \eqref{eq:capTrunc} will be shown to exist in Theorem \ref{theorem:upperCap_anyFSC_limit}.
%Throughout Section~\ref{sec:proof_cap} we also give general lower and upper bounds on any FSC, not only strongly connected FSCs.
%Although the formula in \eqref{eq:capTrunc} is a multi-letter expression with auxiliary RVs that may have an unbounded support, we present two main efficient tools for computing achievable rates for any fixed, finite cardinality for all the auxiliary RVs and any time-invariant function $f(\cdot)$.
%The DP formulation method, introduced in Section~\ref{subsection:prelDP}, is described in the following theorem.
\begin{theorem}[DP Formulation]\label{theorem:MDP}
For a fixed, finite cardinality $\left|\mathcal{U}\right|$ and a function $f: \mathcal U \times \mathcal S \to \mathcal X$, the capacity expression in \eqref{eq:capTrunc} can be formulated as an average-reward DP (Table \ref{table:DP}).
\begin{table}[t]
\caption{The DP formulation} \centering
\label{table:DP}
\begin{tabular}[h]{|l|l|}
\hline 
DP Notations & FSC with Feedback and causal state at the encoder\\
\hline \hline
DP state, $z_{i-1}$& $P(u_{i-1},s_{i-1}|y^{i-1})$ \\
\hline
Disturbance, $w_i$ & $y_i$ - the channel output \\
\hline
Action, $a_i$ & $P(u_i|u_{i-1},y^{i-1})$ \\
\hline
DP state evolution, $z_i=F(z_{i-1},a_{i},w_{i})$  & Eq. \eqref{eq:evolution1}\\
\hline
Reward, $g(z_{i-1},a_i)$ & $I(U_i,U_{i-1};Y_i|y^{i-1})$\\
\hline 
\end{tabular}
\end{table}
\end{theorem}
\begin{figure*}[!b]
\begin{center}
\hrulefill
\begin{equation}
P(u_i,s_i|y^i)
 = \frac{\sum_{u_{i-1},s_{i-1}}\beta_{i-1}(u_{i-1},s_{i-1})a_i(u_i,u_{i-1},y^{i-1})P_{Y,S^+|X,S}(y_i,s_i|f(u_i,s_{i-1}),s_{i-1})}{\sum_{u_{i-1},u'_i,s_{i-1}}\beta_{i-1}(u_{i-1},s_{i-1})a_i(u_i',u_{i-1},y^{i-1})P_{Y|X,S}(y_i|f(u'_i,s_{i-1}),s_{i-1})} \label{eq:evolution1}
\end{equation}
\end{center}
\end{figure*}
%The second tool we suggest for computing achievable rates is established on the $Q$-graph technique.
The proof of Theorem \ref{theorem:MDP} is given in Section~\ref{sec:dpQgraph}. 
%Due to the fact that in the general case, $|\cU|$ in Theorem \ref{theorem:capTrunc} may be unbounded, Theorem \ref{t}
Theorem \ref{theorem:MDP} serves as a tool to compute achievable rates for the feedback capacity either numerically or analytically. For example, 
we can use the value iteration algorithm (VIA), or solve the corresponding Bellman equation. In cases that there is a cardinality bound $|\cU_i|=|\cU|<\infty$, the DP characterizes the feedback capacity itself, and such examples will be given in Section~\ref{sec:examples}.
%Otherwise, i.e., if there is no cardinality bound, fixing $|\cU|< \infty$ and taking the supremum in \eqref{eq:supliminfcapCSI} also with respect to a choice of $f$ \textbf{will provide a lower bound on the capacity, not an equality}. Therefore, in the general case, Theorem \ref{theorem:MDP} serves as a lower bound on the capacity.
%This Theorem provides computable lower bounds on the capacity since $|\cU|$ is finite.

The next theorem provides another tool for computing lower bounds based on the $Q$-graph method. It provides a single-letter lower bound on the feedback capacity for any choice of a $Q$-graph. Given a $Q$-graph, cardinality $|\cU|$, a function $f: \mathcal U \times \mathcal S \to \mathcal X$ and conditional distribution $P(u^+|u,q)$, there is a transition matrix $P(s^+,u^+,q^+|s,u,q)$ given by
\begin{align}
\label{eq:suq_transition}
P(s^+,u^+,q^+|s,u,q) &= \sum_{x,y}  P(s^+,u^+,q^+,x,y|s,u,q) \nonumber \\
&=\sum_{x,y} P(u^+|u,q)\mathbbm{1}\{x=f(u^+,s)\} \mathbbm{1}\{ q^+=g(q,y) \} P_{Y,S^+|X,S}(y,s^+|x,s).
\end{align}
%The single-letter lower bound is given by the following:
%Having this transition matrix a unique stationary distribution $\pi(s,u,q)$, we have a single-letter lower bound on the capacity as presented in the following Theorem. 
%Given a $Q$-graph, a new graph that includes the channel characterization is constructed. The new graph is termed an $(S,U,Q)$-graph.
%The triplet $(S,U,Q)$ corresponds to before a single transmission; $U^+$ determined by $P(u^+|u,q)$ sets the transmitted input symbol $X$ according to $f(u^+,s)$, and the pair $(S^+,Q^+)$ corresponds to after the transmission. 
\begin{theorem}[$Q$-graph Lower Bound]
 \label{theorem:qgraph_LB}
For any Q-graph, given a fixed, finite cardinality $\left|\mathcal{U}\right|$ ($U^+,U \in \cU$) and a function $f: \mathcal U \times \mathcal S \to \mathcal X$, the feedback capacity is lower bounded by
\begin{align}\label{eq:Theorem_Lower}
C_{\text{fb-csi}}\geq I(U^+,U;Y|Q),
\end{align}
for all $P(u^+|u,q)\in\mathcal{P}_\pi$ that are BCJR-invariant, where the joint distribution is 
\begin{align}
\label{eq:Qgraph_joint_dist}
  P(s,u,q,x,y,s^+,u^+,q^+)= \pi(s,u,q) P(u^+|u,q) \mathbbm{1}\{x = f(u^+,s)\} P(y,s^+|x,s) \mathbbm{1}\{q^+ = g(q,y)\}.
\end{align}
\end{theorem}
The notation $\mathcal{P}_{\pi}$ denotes the set of distributions $P(u^+|u,q)$ that induce a transition matrix $P(s^+,u^+,q^+|s,u,q)$ with a unique stationary distribution on the $(S,U,Q)$-graph, denoted by $\pi(s,u,q)$. 
%An input distribution $P(u^+|u,q) \in \mathcal{P}_\pi$ is said to be an \textit{aperiodic} if its corresponding $(S,U,Q)$-graph is aperiodic.
An input distribution is \textit{BCJR-invariant} if the Markov chain $(S^+,U^+)-Q^+-(Q,Y)$ holds. The proof of Theorem \ref{theorem:qgraph_LB} is given in Section~\ref{sec:dpQgraph}. 
%The following lemma, which is also proved in in Section~\ref{sec:dpQgraph}, will help us derive the capacity of the stochastic POST channel.
%The following lemmata are consequences of applying Theorem \ref{theorem:MDP} and Theorem \ref{theorem:qgraph_LB} to known problems that are special cases of the general setting described in Fig \ref{fig:setting}.

The $Q$-graph and the DP methods are detailed in Section~\ref{sec:dpQgraph}, and in Section~\ref{sec:examples} we demonstrate their usefulness on several examples. 
\section{Capacity Computation}
In this section, we elaborate on the DP and $Q$-graph methods of computing the feedback capacity, and prove Theorems \ref{theorem:MDP} and \ref{theorem:qgraph_LB}.

\label{sec:dpQgraph}
\subsection{Dynamic Programming Formulation and Proof of Theorem \ref{theorem:MDP}}
%In this paper, we are going to formulate a lower bound on the capacity as an infinite-horizon average-reward dynamic program.
%For a finite, fixed cardinality $|\cU|$ and fixed $f: \mathcal U \times \mathcal S \to \mathcal X$, we formulate the capacity expression \eqref{eq} a DP
%The proof of Theorem \ref{theorem:MDP} follows from characterizing the capacity expression \eqref{eq:capTrunc} with some choice of cardinality $|\cU|< \infty$ and a function $f(\cdot)$ as the optimal average reward of 
In words, the DP summarized in Table \ref{table:DP} is as follows. The DP state, $z_{i-1}$, is chosen as the conditional joint distribution matrix whose elements are  $\beta_{i-1}(u_{i-1},s_{i-1})= P(u_{i-1},s_{i-1}|y^{i-1})$, for $u_{i-1},s_{i-1} \in \cU \times \cS$. The action space, $\mathcal{A}$, is the set of stochastic matrices $P(u_i|u_{i-1},y^{i-1})$. The disturbance is taken to be the channel output, i.e., $w_i=y_i$. Finally, the reward function at time $i$ is $I(U_i,U_{i-1};Y_i|y^{i-1})$. Now in order to prove that this is a valid DP formulation, we need to prove the following lemma.
\begin{lemma}[DP formulation]
\label{lem:DP_formulation}
The formulation presented in Table \ref{table:DP} satisfies the DP model. That is,
\begin{enumerate}
    \item The DP state is a time-invariant function of the previous DP state, action and disturbance. 
    \item The disturbance is conditionally dependent on the DP state and action.
    \item The reward is a time-invariant function of the state and action.
\end{enumerate}
\end{lemma}
Consequently, the optimal average reward is
\begin{align}
&\rho^*= \sup_{\pi} \liminf_{N\to\infty} \frac{1}{N} \sum_{i=1}^{N} I_\pi (U_i,U_{i-1};Y_i|Y^{i-1}), \nonumber
\end{align}
%%%%%%%
where the subscript $\pi$ indicates that the mutual information corresponds to the policy denoted by $\pi$. The proof of Theorem \ref{theorem:MDP} is a direct consequence of Lemma \ref{lem:DP_formulation} and the following Lemma \ref{lem:LB_supliminf}, whose proofs are given in Appendices \ref{appendix:lem_DP_proof} and \ref{appendix:LB_supliminf}, respectively.
%Here, we focus on Theorem \ref{theorem:MDP} to compute lower bounds on the feedback capacity 
%In this section, we characterize the capacity expression \eqref{eq:capTrunc} with a finite, fixed cardinality $|\cU|$ and fixed $f: \mathcal U \times \mathcal S \to \mathcal X$ as the optimal average reward of a DP. 
%To conclude the proof, for a finite cardinality $\left|\mathcal{U}\right|$ and a time-invariant function $f: \mathcal U \times \mathcal S \to \mathcal X$, \eqref{eq:achievability} is a lower bound to the capacity in \eqref{eq:capDI}, and it can be formulated as an infinite-horizon average-reward DP.
\begin{lemma}
\label{lem:LB_supliminf}
The feedback capacity of a connected FSC with SI known causally at the encoder can be expressed by
\begin{equation}
\label{eq:supliminfcapCSI}
    C_{\text{fb-csi}} = \sup \liminf_{N\to\infty}   \frac{1}{N}  \sum_{i=1}^{N}  I(U_i,U_{i-1};Y_i|Y^{i-1}) \text{,}
\end{equation}
where the supremum is taken with respect to ${\{P(u_i|u_{i-1},y^{i-1})\}}_{i\geq1}$, and the joint distribution is given by \eqref{eq:joint_distTrunc}.
%and a time-invariant function $f:\cU \times \cS \to \cX$, 
\end{lemma}
We note that if the feedback capacity can be achieved with $|\cU_i|=|\cU|<\infty$,
%all alphabets $\cU_i$ have a finite cardinality bound $|\cU|$ that does not depend on $i$,
the supremum in \eqref{eq:supliminfcapCSI} is also taken with respect to a time-invariant function $f:\cU \times \cS \to \cX$.
%since it is sufficient to use only a part of the total $|\cX|^{{|\cS|}^i}$ distinct functions $\{f_u (s):\cS \to \cX \}$ until time $i$.
%It then follows that $\rho^*$ equals the achievable rate shown in Theorem \ref{lem:LB_supliminf}.
%%%%%%%%%%
\subsection{Q-Graph Technique and Proof of Theorem \ref{theorem:qgraph_LB}}
\label{sec:qgraph}
An altrenative method to compute lower bounds on the feedback capacity, besides the DP, is the $Q$-graph technique in Theorem \ref{theorem:qgraph_LB}. 
%, i.e., we derive a single-letter lower bound on the capacity of the general setting using the $Q$-graph technique. This lower bound serves as an alternative for computing lower bounds via the DP. Further, we formalize this lower bound as an optimization problem.
Here, we explain this theorem, and prove it in the next part of this section. The main idea of the proof is to embed an auxiliary graph into the capacity expression in \eqref{eq:capTrunc}. Throughout this section, we assume a fixed cardinality $|\cU|$ and a fixed function $f: \mathcal U \times \mathcal S \to \mathcal X$.
%%%%%%%%%%
% ************ COMPLEXITY!!! *************
% As mentioned, one can apply the VIA for the DP tool to compute numerical achievable rates. However, the complexity is subject to the chosen cardinality $|\cU|$. For instance, assume $|\cS|=2$ and identical quantization of $n$ points for both the DP and action spaces. Taking $|\cU|=2$ implies $|\cS|\times|\cU|-1=3$ DP state spaces and $|\cU|\times (|\cU|-1)=2$ action spaces; this results in complexity of $O(n^5)$, which is tractable. Further, taking $|\cU|=3$ implies $5$ DP state spaces and $6$ action spaces, resulting in complexity of $O(n^{11})$. In other words, for $|\cU|>2$, the VIA complexity becomes intractable. On the other hand, the $Q$-graph optimization may work properly on a machine even for larger $|\cU|$.
% ***********
%%%%%%%%%

The FSC $P_{Y,S^+|X,S}$ is embedded into a given $Q$-graph by constructing a new directed, connected graph termed an \textit{$(S,U,Q)$-graph} to include the information on the $Q$-graph and on the evolution of the state and the auxiliary RV pair, $(S,U)$. The $(S,U,Q)$-graph is constructed as follows:
\begin{enumerate}
    \item Each node in the $Q$-graph is split 
    to $|\cS|\times |\cU|$ new nodes represented by $(s,u,q)\in\cS\times\cU\times\cQ$.
    \item An edge $(s,u,q)\rightarrow(s^+,u^+,q^+)$, with a label $(x,y)$, exists if and only if there exists a pair $(x,y)$ such that $x=f(u^+,s)$, $P_{Y,S^+|X,S}(y,s^+|x,s)>0$ and $q^+=g(q,y)$. 
\end{enumerate}
%%%%%%%%%%%%%%%%%
% For any input distribution $P(u^+|u,q)$ the transition probabilities in the $(S,U,Q)$-graph are given by 
% \begin{align}
% \label{eq:suq_transition}
% P(s^+,u^+,q^+|s,u,q) &= \sum_{x,y}  P(s^+,u^+,q^+,x,y|s,u,q) \nonumber \\
% &=\sum_{x,y} P(u^+|u,q)\mathbbm{1}\{x=f(u^+,s)\} \mathbbm{1}\{ q^+=g(q,y) \} P(y,s^+|x,s).
% \end{align}
%%%%
%For each input distribution $P(u^+|u,q)$, \eqref{eq:BCJR_ev} can be written as a mapping $B_{u,s}:(\cS \times \cU -1)\times \cY \to [0,1]$.
We denote by $\mathcal{P}_{\pi}$ the set of input distributions $P(u^+|u,q)$ that induce a unique stationary distribution on $(S,U,Q)$, i.e., their corresponding $(S,U,Q)$-graph is irreducible and aperiodic. An input distribution $P(u^+|u,q)$ is said to be \textit{aperiodic} if its $(S,U,Q)$-graph is aperiodic. An aperiodic input distribution is \textit{BCJR-invariant} if it induces the Markov chain
\begin{equation}
(S^+,U^+)-Q^+-(Q,Y).    
\end{equation}

For the proof of Theorem \ref{theorem:qgraph_LB}, we use similar ideas as in the proof of \cite[Theorem 3]{Sabag_UB_IT} based on our Lemma~\ref{lem:LB_supliminf}. We show that a BCJR-invariant input distribution induces for all $i$ the Markov chain $Y_i-Q_{i-1}-Y^{i-1}$, which leads to the fact that $I(U^+,U;Y|Q)$ with the chosen input distribution is a lower bound on the feedback capacity.

For an integer $i$ we define $q_{i}\triangleq \Phi_{i}(y^{i})$ and prove by induction that $P(u_{i},s_{i}|y^{i},q_{i})=\pi(u_{i},s_{i}|q_{i})$ 
for the choice of a BCJR-invariant input distribution. At time $i-1$, assume that $P(u_{i-1},s_{i-1}|y^{i-1},q_{i-1})=\pi(u_{i-1},s_{i-1}|q_{i-1})$. Then, at time $i$ we have 
\begin{align}
P(u_{i},s_{i}|y^{i},q_{i})&\stackrel{(a)}=P(u_{i},s_{i}|q_{i-1},y_i) \nn\\
&\stackrel{(b)}=P(u_{i},s_{i}|q_{i}).
\end{align}
Eq. (a) follows from the fact that the BCJR recursive equation \eqref{eq:evolution1}, for any $u,s \in \cU\times \cS$,
% derived in Appnedix~\ref{appendix:lem_DP_proof}),
% %of the new channel state and auxiliary RV pair
% \begin{align}
% \label{eq:BCJR_ev}
% P(u_i,s_i|y^i)&=  \frac{\sum\limits_{u_{i-1},s_{i-1}}P(u_{i-1},s_{i-1}|y^{i-1})P(u_i|u_{i-1},y^{i-1})P(y_i|f(u_i,s_{i-1}),s_{i-1})P(s_i|f(u_i,s_{i-1}),s_{i-1},y_i)}{\sum\limits_{u_{i-1},u'_i,s_{i-1}} P(u_{i-1},s_{i-1}|y^{i-1})P(u'_i|u_{i-1},y^{i-1})P(y_i|f(u'_i,s_{i-1}),s_{i-1})},
% \end{align}
can be computed from $\{P(u_{i-1},s_{i-1}|y^{i-1})\},\{P(u_i|u_{i-1},y^{i-1})\}$ and $y_i$, while for the first distribution we use the induction hypothesis, and for the second distribution we use the assumption that the inputs are of the form $\{P(u_i|u_{i-1},q_{i-1}(y^{i-1}))\}$. Eq. (b) follows from $q_i=g(q_{i-1},y_i)$ and from the BCJR-invariant property.
%%%
As a result, the Markov chain $Y_i-Q_{i-1}-Y^{i-1}$ holds for all $i$:
\begin{align}
    P(y_i|y^{i-1},q_{i-1})&=\sum_{s_{i-1},u_{i-1},u_i}P(s_{i-1},u_{i-1},u_i,y_i|y^{i-1},q_{i-1}) \nonumber\\
    &=\sum_{s_{i-1},u_{i-1},u_i} P(u_{i-1},s_{i-1}|y^{i-1},q_{i-1}) P(u_i|u_{i-1},s_{i-1},y^{i-1},q_{i-1}) P(y_i|f(u_i,s_{i-1}),s_{i-1}) \nonumber\\
    &\stackrel{(a)}=\sum_{s_{i-1},u_{i-1},u_i}P(u_{i-1},s_{i-1}|y^{i-1},q_{i-1})P(u_i|u_{i-1},q_{i-1})P(y_i|f(u_i,s_{i-1}),s_{i-1}) \nonumber\\
    &\stackrel{(b)}=\sum_{s_{i-1},u_{i-1},u_i}\pi(u_{i-1},s_{i-1}|q_{i-1})P(u_i|u_{i-1},s_{i-1},q_{i-1})P(y_i|f(u_i,s_{i-1}),s_{i-1}) \nonumber\\
    &=P(y_i|q_{i-1}) \label{eq:markovYQ},
\end{align}
where
\begin{enumerate}[label={(\alph*)}]
\item follows from the assumption on the form of the input distribution;
\item follows from the inductive argument shown above.
\end{enumerate}
We now turn to prove Theorem \ref{theorem:qgraph_LB}.
\begin{proof}[Proof of Theorem \ref{theorem:qgraph_LB}]
The proof of the theorem is completed by the following chain of inequalities:
\begin{align}
    C_{\text{fb-csi}}&\stackrel{(a)}= \sup_{\substack{\{P(u_i|u_{i-1},y^{i-1})\}_{i\geq1},\\x_i=f(u_i,s_{i-1})}} \liminf_{N\to\infty}  \frac{1}{N} \sum_{i=1}^N I(U_i,U_{i-1};Y_i|Y^{i-1}) \nonumber\\
    &\stackrel{(b)}\geq \sup_{\substack{\{P(u_i|u_{i-1},y^{i-1})\}_{i\geq1},\\x_i=f(u_i,s_{i-1})}} \liminf_{N\to\infty}  \frac{1}{N} \sum_{i=1}^N I(U_i,U_{i-1};Y_i|Q_{i-1})-I(Y_i;Y^{i-1}|Q_{i-1})\nonumber\\
    &\stackrel{(c)}\geq \liminf_{N\to\infty}  \frac{1}{N} \sum_{i=1}^N I(U_i,U_{i-1};Y_i|Q_{i-1})-I(Y_i;Y^{i-1}|Q_{i-1})\nonumber\\
    &\stackrel{(d)}= \liminf_{N\to\infty}  \frac{1}{N} \sum_{i=1}^N I(U_i,U_{i-1};Y_i|Q_{i-1})\nonumber\\
    &\stackrel{(e)}= I(U^+,U;Y|Q), \label{eq:single-letter}
\end{align}
where
\begin{enumerate}[label={(\alph*)}]
\item follows from Lemma \ref{lem:LB_supliminf}.
\item follows from adding and subtracting $H(Y_i|Q_{i-1})$, and from \\$H(Y_i|U_i,U_{i-1},Y^{i-1})\leq H(Y_i|U_i,U_{i-1},Q_{i-1})$.
\item follows by considering BCJR input distribution $P(u_i|u_{i-1},y^{i-1})=P(u^+|u,q)\in \mathcal{P}_\pi$ for all $i$.
\item follows from the Markov chain $Y_i-Q_{i-1}-Y^{i-1}$.
\item follows from the convergence of Markov chains due to the aperiodic input distribution, and the continuity of the mutual information with respect to the joint distribution.
%follows from the stationary distribution induced by the aperiodic input distribution.
% follows from the aperiodic Markov chain on the state space S × Q) which induces its corresponding stationary distribution.
\end{enumerate}
\end{proof}

\section{Examples}
\label{sec:examples}
In this section, we demonstrate the DP and the $Q$-graph methods analytically and numerically for several examples.

\subsection{Look-Ahead State-Dependent Channels}
\begin{figure}[t]
\begin{center}
\begin{psfrags}
    \psfragscanon
    \psfrag{E}[][][1]{$M$}
    \psfrag{S}[][][1]{$\quad S^{i-1+l}$}
    %{\\}
    \psfrag{A}[\hspace{2cm}][][1]{Encoder}
	 \psfrag{F}[\hspace{1cm}][][1]{$X_i$}
	 \psfrag{B}[\hspace{2cm}]{{$P(y_i|x_i,s_{i-1})$}}
	 \psfrag{G}[][][1]{$Y_i$}
	 \psfrag{C}[\hspace{2cm}][][1]{Decoder}
	 \psfrag{K}[][][1]{$\hat{M}$}
	 \psfrag{H}[\hspace{2cm}][][1]{$Y_i$}
	 \psfrag{D}[\hspace{2cm}][][1]{Unit-Delay}
	 \psfrag{J}[\vspace{2cm}\hspace{2cm}][][1]{$Y_{i-1}$}
	 \psfrag{L}[\hspace{2cm}][][1]{Finite-State Channel}
	 \psfrag{I}[][][1]{}
	 %\psfrag{tag}[][][<scale>]{Latex Text}
\includegraphics[scale=0.8]{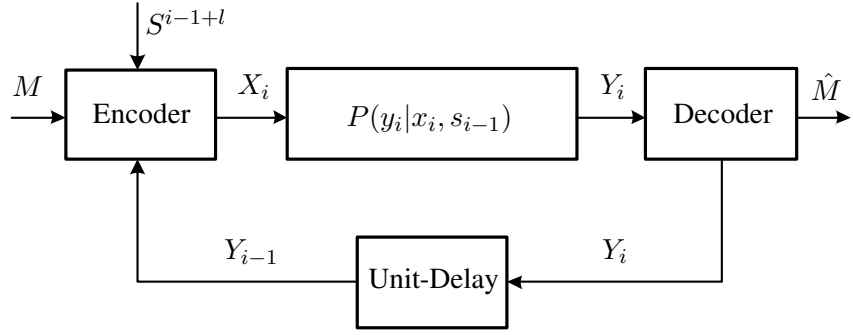}
\caption{
A state-dependent channel where the state is available at the encoder with a finite look-ahead $l>0$ in the presence of feedback. Note that the current state at time $i$ is $S_{i-1}$.}
\label{fig:look-ahead}
\psfragscanoff
\end{psfrags}
\end{center}
\end{figure}

\begin{figure}[b]
\begin{center}
\begin{psfrags}
    \psfragscanon
    \psfrag{A}[][][0.8]{$0$}
    \psfrag{B}[][][0.8]{$1$}
    \psfrag{C}[][][0.8]{$0$}
    \psfrag{D}[][][0.8]{$1$}
    \psfrag{E}[][][1]{$s_{i-1}=0$}
    \psfrag{F}[][][1]{$x_i$}
    \psfrag{G}[][][1]{$y_i$}
    \psfrag{H}[][][0.8]{$1$}
    \psfrag{I}[][][0.8]{\raisebox{-0.4cm}{$0.5$\hspace{-0.1cm}}}
    \psfrag{J}[][][0.8]{$0.5$\hspace{-0.4cm}}
    %
    % \psfrag{E1}[][][1]{$Q(s_i|y_i)$}
    % \psfrag{F1}[][][1]{$y_i$}
    % \psfrag{G1}[][][1]{$s_i$}
    % \psfrag{H1}[][][0.8]{$1$}
    % \psfrag{I1}[][][0.8]{\raisebox{-0.4cm}{$1-\eta$\hspace{-0.1cm}}}
    % \psfrag{J1}[][][0.8]{$\eta$\hspace{-0.4cm}}
    % %
    \psfrag{K}[][][0.8]{$0$}
    \psfrag{L}[][][0.8]{$1$}
    \psfrag{M}[][][0.8]{$0$}
    \psfrag{N}[][][0.8]{$1$}
    \psfrag{O}[][][1]{$s_{i-1}=1$}
    \psfrag{P}[][][1]{$x_i$}
    \psfrag{Q}[][][1]{$y_i$}
    \psfrag{R}[][][0.8]{$0.5$}
    \psfrag{S}[][][0.8]{\raisebox{-0.4cm}{$1$}}
    \psfrag{T}[][][0.8]{$0.5$\hspace{0.6cm}}
\includegraphics[scale=0.28]{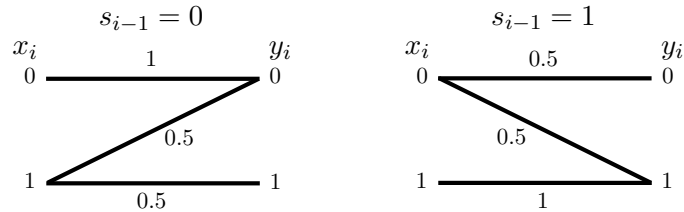}
\caption{The $ZS$ channel topology for $P_{Y|X,S}(y_i|x_i,s_{i-1})$.} \label{fig:ZS-channel}
\psfragscanoff
\end{psfrags}
\end{center}
\end{figure}
Consider DMCs $P_{Y|X,S}(y_i|x_i,s_{i-1})$ with memoryless states that are distributed according to $P(s)$. When the states are
available at the encoder causally or non-causally, 
the capacity is achieved by the Shannon strategy (SS) \cite{Shannon58} or by the Gel'fand–Pinsker (GP) coding scheme \cite{GePi80,HeegardElGamal_state_encoded83}, respectively. An intermediate situation that may occur in practice is when the states are known at the encoder with a finite LA. That is, at time $i$ (with current state $s_{i-1}$), the encoder has access to the states sequence $(s_0,\dots,s_{i-1+l})$, where $l>0$ is a finite LA parameter, as illustrated in 
Fig. \ref{fig:look-ahead}. The capacity of this problem is unknown as a computable expression, and only sequences of lower and upper bounds are known \cite{weissman2006source}. 

In the LA setting, the access of the encoder to future states can be reformulated as a causal access to states that are Markovian. For instance, if $l=1$, at time $i$ the encoder knows $s^{i}$, that is, it has access to the future state $s_i$. If we define a new state as the tuple $\tilde{s}_{i-1}={s}_{i-1}^{i}$, we obtain a FSC characterized by $P_{Y,\tilde{S}^+|X,\tilde{S}}(y_i,\tilde{s}_{i-1}|x_i,\tilde{s}_{i-2})=P_{Y|X,\tilde{S}}(y_i|x_i,\tilde{s}_{i-2})P_{\tilde{S}^+|\tilde{S}}(\tilde{s}_{i-1}|\tilde{s}_{i-2})$, that is, the states are Markovian (rather than memoryless), since both $\tilde{s}_{i-1}$ and $\tilde{s}_{i-2}$ share the element $s_{i-1}$. This transformation can be trivially extended to $l>1$. The advantage of this formulation is that Markovian states are a special case of our general FSC setting.
%the DMC with LA SI $\tilde{S}^{i-1+l}$ known at the encoder and feedback, i.e., where the encoder knows all the states $l$ time epochs ahead the current state $\tilde{S}_{i-1}$ (see Fig. \ref{fig:look-ahead}), is also a special case of the FSC in Fig. \ref{fig:setting} with state $S_{i-1}$. This can be shown straightforwardly by identifying . 
% %, not i.i.d., but we notice that at time $i$ the encoder can transmit an input symbol $x_i$ that depends on $s_{i-1}$ which is equivalent to 
% %the futuristic states of the DMC $\tilde{S}_{i-1}^{i-1+l}$ (assuming $l\ge 1$)$
% $S_{i-1}= in a non-causal manner.
% %and the dependency of $P(y_i|x_i,s_{i-1-l})$ in $s_{i-1-l}\triangleq (\tilde{s}_{i-1-l},\dots, \tilde{s}_{i-1})$ is only via the right-hand component $\tilde{s}_{i-1}$. 
% %The key to the proof is to consider $S_{i-1}=\tilde{S}_{i-1}^{i-1+l}$ (thus $\cS= \tilde{\cS}^{l+1}$); the full proof is given in Appendix ****.
%Clearly, when $l=0$ the state is known causally, i.e., Shannon's setting is obtained, and when $l=n-i$ ($n\to \infty$ is the block size) the encoder knows all the states non-casually, which is another classic setting covered by Gel'fand and Pinsker \cite{GePi80}.
% \textbf{can come later:
% While it is known that for causal/non-causal SI feedback does not increase the capacity, it has remained an open problem whether feedback increases the capacity of look-ahead.} 
%This setting without feedback was introduced in \cite{weissman2006source}.

We consider a state-dependent channel with binary state $S_i\stackrel{\text{i.i.d.}}\sim$ Bernoulli($0.5$) and $|\cX| = |\cY| = 2$, where the output depends on the input and the state according to the $ZS$-channel topology in Fig.~\ref{fig:ZS-channel}. 
We thus call this channel \textit{the i.i.d. $ZS$-channel}. For $l=1$, we derive a closed-form lower bound on its feedback capacity utilizing the $Q$-graph method in Theorem~\ref{theorem:qgraph_LB}.
%; i.e., if the input and the channel state agree, then the output is equal to them, otherwise, the output is a random instance. 

\begin{theorem}
\label{theorem:LB-ZS-LA}
The feedback capacity of the i.i.d. $ZS$-channel when the state is available at the encoder with a single LA, $C_{\text{FB-LA-}1}$, is lower-bounded by
\begin{align}
    C_{\text{FB-LA-}1} \ge \textstyle \frac{1}{2} \left(1 - H(\frac{1}{8})\right)=\frac{7}{16} \log(7)-1 \approx 0.228217.     \label{eq:lb-zs-la}
\end{align}
% \begin{align}
%     C_{\text{CSI-E}}&= 1-H(\textstyle\frac{1}{4}) \approx 0.188722 \nn\\
%     &\le R_{\text{FB-LA-}1}= \textstyle \frac{1}{2} (1 - H(\frac{1}{8}))=\frac{7}{16} \log(7)-1 \approx 0.228217 \nn\\
%     &\le  C_{\text{FB-LA-}1} \le C_{\text{SI-E}} \approx 0.271553. \label{eq:lb-zs-la}
% \end{align}
\end{theorem}
% \begin{table}[b]
% \caption{The feedback capacity of the i.i.d. $ZS$ channel with $1$ look-ahead SI known at the encoder compared with the capacities $C$, $C_{\text{CSI-E}}$, $C_{\text{SI-E}}$ and $C_{\text{SI-ED}}$\\.}
% \centering
% \label{table:BoundsLA}
% \begin{tabular}[h]{|c|c|c|c|}
% \hline
% $C = C_{\text{CSI-E}}$ & $R_{\text{FB-LA-}1}$ & $C_{\text{SI-E}}$ & $C_{\text{SI-ED}}$\\
% \hline
% $1-H(\textstyle\frac{1}{4}) \approx 0.188722$ & $\textstyle \frac{1}{2} (1 - H(\frac{1}{8}))=\frac{7}{16} \log(7)-1 \approx 0.228217$ &  $\sim 0.271553$ & $\log (\frac{5}{4}) \approx
% 0.321928$ \\
% \hline
% \end{tabular}
% \end{table}

\begin{table}[t]
\caption{Our achievable rate for the i.i.d. $ZS$-channel with $1$ LA and feedback, $R_{\text{FB-LA-}1}$, compared with the capacity of various scenarios of SI availability at the encoder and the decoder.}
%$1$ look-ahead SI known at the encoder compared with the capacities $C$, $C_{\text{CSI-E}}$, $C_{\text{SI-E}}$ and $C_{\text{SI-ED}}$\\.}
\centering
\label{table:BoundsLA}
\begin{tabular}[h]{|c|c|c|c|}
\hline
\textbf{Encoder's SI} & \textbf{Decoder's SI} & \textbf{Rate} \\
\hline
-- & -- & $C=1-H(\textstyle\frac{1}{4}) \approx 0.188722$\\
\hline
Causal ($l=0$) & -- & $C_{\text{CSI-E}}=1-H(\textstyle\frac{1}{4}) \approx 0.188722$\\
\hline
Single LA ($l=1$) & -- & $R_{\text{FB-LA-}1}=\textstyle \frac{1}{2} (1 - H(\frac{1}{8}))=\frac{7}{16} \log(7)-1 \approx 0.228217$ \\
\hline
Non-causal ($l=\infty$) & -- & $C_{\text{SI-E}}\approx 0.271553$\\
\hline
Causal / Non-causal & Causal / Non-causal & $C_{\text{SI-ED}}=\textstyle H(\frac{1}{5})-\frac{2}{5}=\log (\frac{5}{4}) \approx 0.321928$ \\
\hline
\end{tabular}
\end{table}
% \begin{align}
%     C_{\text{CSI-E}}&= 1-H(\textstyle\frac{1}{4}) \approx 0.188722 \nn\\
%     &\< R_{\text{FB-LA-}1}= \textstyle \frac{1}{2} (1 - H(\frac{1}{8}))=\frac{7}{16} \log(7)-1 \approx 0.228217 \nn\\
%     &\le  C_{\text{FB-LA-}1} \le C_{\text{SI-E}} \approx 0.271553. \label{eq:bounds-zs-la}
% \end{align}

The proof of Theorem~\ref{theorem:LB-ZS-LA} is given below and relies on a particular choice of policy in Theorem~\ref{theorem:qgraph_LB}. In Table \ref{table:BoundsLA}, we compare the achievable rate in Theorem \ref{theorem:LB-ZS-LA}, denoted by $R_{\text{FB-LA-}1}$, with the capacities of the SS ($l=0$), the GP \footnote{For evaluating $C_{\text{SI-E}}$, which is concave, we programmed a code that utilizes CVX, a MATLAB-based modeling system for convex optimization. Our code is available online in \url{https://github.com/Eli-BGU/Gelfand-Pinsker-capacity-computation/}.} ($l=\infty$), the scenario when the state is not available at either the encoder or the decoder and the scenario when it is available at the both parties. It can be observed that causal SI ($l=0$) known at the encoder does not increase the capacity of this channel (see a detailed proof in Appendix~\ref{appendix:CSInotIncreaseZScapacity}). However, we note that even a single LA ($l=1$) increases the feedback capacity by at least $\frac{R_{\text{FB-LA-}1}-C_{\text{CSI-E}}}{C_{\text{SI-E}}-C_{\text{CSI-E}}}\times 100\approx 48\%$ compared to $C_{\text{SI-E}}$ ($l=\infty$). 

We recall that feedback does not increase the capacity for the scenarios when the SI is available causally/noncausally at the encoder (see, e.g., \cite[Prob.~17.17]{el2011network}). However, it is still unknown whether feedback increases the capacity of LA state-dependent channels. Our result does not provide an answer to this open problem, but we suspect that the answer is yes, since the feedback is necessary in the derivation of $R_{\text{FB-LA-}1}$ in the proof of Theorem \ref{theorem:LB-ZS-LA}. That is, the policy is feedback-dependent and relies on a $2$-nodes $Q$-graph. If there exists an upper bound on the non-feedback LA capacity, which is smaller than $R_{\text{FB-LA-}1}$, it would establish our claim. The upper bound given in \cite{weissman2006source} is represented as a function of a chosen parameter. To obtain meaningful upper bounds, the parameter should be chosen large, but this comes at the expense of infeasible computability. 
% However, the information rates achieved from each node are equal (??), so we suspect feedback actually does not increase the capacity of our setting.
% If this is true, it will imply that feedback increases the capacity of state-dependent channels with finite LA. However, our results do not guarantees this claim yet, since, to the best of our knowledge, the capacity is unknown as a computable closed-form expression in the literature. % \textbf{I wrote this paragraph above at a higher level.}
% Notice that Nodes $Q=1$ and $Q=2$ in Fig. \ref{fig:Qgraph_2node} are associated with feedback outputs $Y=0$ and $Y=1$, respectively. We also note that the conditional distributions $P(u^+|u,q)$ of these two nodes, specified in the proof below, are unequal, i.e., $P(u^+|u,Q=1)\ne P(u^+|u,Q=2)$. Thus, Nodes $Q=1$ and $Q=2$ cannot be merged, and the output feedback is crucial in deriving $R_{\text{FB-LA-}1}$, and 

\begin{remark}
$R_{\text{FB-LA-}1}$ relies on an auxiliary RV with $|\cU|=4$, and a particular choice of $x=f(u^+,\tilde{s})$. However, we were not able to improve this achievable rate numerically. In particular, we increased $|\cU|$, optimized $x=f(u^+,\tilde{s})$, and evaluated the lower bound $Q$-graphs (up to size $|\cQ|=7$), but all led to the same lower bound on the feedback capacity. These simulations may indicate that $R_{\text{FB-LA-}1}$ is actually the feedback capacity for $l=1$, but we do not have a matching converse.
\end{remark}

\begin{remark}
As mentioned, causal SI does not increase the capacity in this example. This fact is also reflected from the policy used in the proof of Theorem \ref{theorem:LB-ZS-LA} with a function $x_i=f(u_i,\tilde{s}_{i-1})$ that only depends on $s_i$. That is, the optimal strategy function is independent of $s_{i-1}$.
%since the best function $x_i=f(u_i,\tilde{s}_{i-1})$ is independent of $s_{i-1}$ and depends on $s_i$ only.
% \textbf{to delete?}
% The fact that causal CSI does not increase the capacity is also reflected in the choice of $x=f(u^+,s)$ specified in the proof in \eqref{eq:fLA}, i.e., $f(u^+,\tilde{S}=0)=f(u^+,\tilde{S}=2)$ and $f(u^+,\tilde{S}=1)=f(u^+,\tilde{S}=3)$. Recalling that $\tilde{S}_{i-1}=(S_{i-1},S_{i})$, notice that both $\tilde{s}_{i-1}\in \{0,2\}$ have the same component $s_i=0$, and both $\tilde{s}_{i-1}\in \{1,3\}$ have the same component $s_i=1$. That is, $f(u^+,\tilde{S}=0)$ is ignorant of the left-hand component $S_{i-1}$ and it depends only on the right-hand component $S_i$.
\end{remark}

% capacity expressions (see, e.g., \cite{el2011network}), denoted by $C_{\text{CSI-E}}$ and $C_{\text{SI-E}}$, respectively. 
%Table \ref{table:BoundsLA} compares it (denoted by $R_{\text{FB-LA-}1}$) with the capacity of this channel under various scenarios of SI availability, including $C_{\text{CSI-E}}$ and $C_{\text{SI-E}}$. 
% \begin{table}[t]
% \caption{The transformation $S_{i-1}=(\tilde{S}_{i-1},\tilde{S}_{i})$ between the binary states of a DMC with a LA of $l=1$ to a quaternary state of the FSC given in Fig. \ref{fig:setting}:}
% \centering
% \label{table:stateLA}
% \begin{tabular}[h]{||c||c|c||}
% \hline \hline
% $s_{i-1}$ & $\tilde{s}_{i-1}$ & $\tilde{s}_i$ \\
% \hline \hline
% $0$ & $0$ & $0$ \\
% \hline
% $1$ & $0$ & $1$ \\
% \hline
% $2$ & $1$ & $0$ \\
% \hline
% $3$ & $1$ & $1$ \\
% \hline \hline
% \end{tabular}
% \end{table}

\begin{figure}[b]
\begin{center}
\begin{psfrags}
    \psfragscanon
    \psfrag{A}[\hspace{2cm}][][0.8]{$Q=1$}
    \psfrag{B}[\hspace{2cm}][][0.8]{$Q=2$}
    \psfrag{C}[][][0.8]{$Y=1$\hspace{0.5cm}}
    \psfrag{D}[\hspace{2cm}][][0.8]{$Y=0$\hspace{-0.5cm}}
    \psfrag{E}[\hspace{2cm}][][0.8]{$Y=0$}
    \psfrag{F}[\hspace{2cm}][][0.8]{$Y=1$}
\includegraphics[scale=0.4]{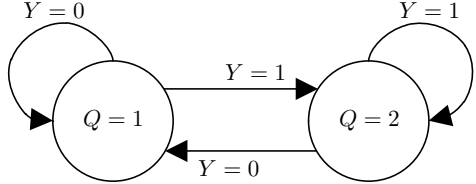}
\caption{A $2$-node $Q$-graph where each node corresponds to $Y=0$ or $Y=1$. %the noisy-Ising($\eta$) channel, $R_{\text{Analytic}}(\eta)$.
} \label{fig:Qgraph_2node}
\psfragscanoff
\end{psfrags}
\end{center}
\end{figure}
\begin{proof}[Proof of Theorem \ref{theorem:LB-ZS-LA}]
The proof follows from specifying the equivalent FSC with SI known causally at the encoder and feedback, then utilizing the $Q$-graph lower bound in Theorem \ref{theorem:qgraph_LB}.
%with Markovian states $\tilde{\cS}$) of the original DMC (with memoryless states $\cS$),
Since $l=1$, we have FSC states $\tilde{\cS}=\{0,1,2,3\}$
corresponding to pairs of consecutive memoryless states
$(s_{i-1},s_{i})\in \{(0,0),(0,1),(1,0),(1,1)\}$, respectively.
%and the obtained FSC is $P(y_i|x_i,\tilde{s}_{i-2})P(\tilde{s}_{i-1}|\tilde{s}_{i-2})$ 
%Applying the lower bound of Theorem \ref{theorem:qgraph_LB} with the $2$-node $Q$-graph in Fig. \ref{fig:Qgraph_2node}, $|\cU|=4$, $x=f(u^+,s)$ and matrices $P(u^+|u,Q=q)$ specified below, the lower bound in \eqref{eq:lb-zs-la} is derived as follows. 
%$|\tilde{\cS}|=4$, where $\tilde{S}_{i-1}=(S_{i-1},S_{i})$, i.e., $\tilde{\cS}=\{00,01,10,11\}$.
%The equivalent FSC $P_{Y|X,\tilde{S}}(y_i|x_i,\tilde{s}_{i-1-l})P_{\tilde{S}^+|\tilde{S}}(\tilde{s}_{i-l}|\tilde{s}_{i-1-l})$ 
%is specifically characterized by the transition matrix $P_{\tilde{S}^+|\tilde{S}}$, whose rows and columns represent the current and next states, respectively, given by
%%%%%%%%%%%%%%%
% \begin{align}
% P_{\tilde{S}^+|\tilde{S}}(\tilde{s}^+|\tilde{s})=
% \begin{array}{@{}c@{}}
%     \rowind{$\tilde{s}=0$} \\ \rowind{$\tilde{s}=1$} \\ \rowind{$\tilde{s}=2$} \\ \rowind{$\tilde{s}=3$} 
%   \end{array}
%   \mathop{\left[
%   \begin{array}{ *{4}{c} }
%      \colind{0.5}{$\tilde{s}^+=0$}  &  \colind{0.5}{$\tilde{s}^+=1$}  &  \colind{0}{$\tilde{s}^+=2$}  & \colind{0}{$\tilde{s}^+=3$} \\
%     0 &  0  &  0.5  & 0.5 \\
%      0.5  & 0.5 &  0  & 0 \\
%      0  &  0  & 0.5 & 0.5 \\
%   \end{array}
%   \right]}^{
%   \begin{array}{@{}c@{}}
%     %\rowind{State} \\ \mathstrut
%   \end{array}
%   }.
% \end{align}
%%%%%%%%%%%%%%5
% \begin{align}
% P(\tilde{s}_{i}|\tilde{s}_{i-1})=
% \begin{bmatrix}
% 0.5 & 0.5 & 0 & 0 \\
% 0 & 0 & 0.5 & 0.5 \\
% 0.5 & 0.5 & 0 & 0 \\
% 0 & 0 & 0.5 & 0.5  \\
% \end{bmatrix}.    \nn
% \end{align}
Consider the $Q$-graph in Fig. \ref{fig:Qgraph_2node}, $|\cU|=4$, a function $f(u^+,\tilde{s})$ given by the following matrix whose rows and columns represent $u^+$ and $s$, respectively, and matrices $P(u^+|u,Q=q)$ whose rows and columns represent $u$ and $u^+$, respectively:
%$Q$-graph method to derive lower bounds on the original DMC with a LA. 
% \begin{blockarray}{ccccc} 
%      \scriptstyle x=f(u^+,\tilde{s})  & \scriptscriptstyle \tilde{s}=0 & \scriptscriptstyle \tilde{s}=1 & \scriptscriptstyle \tilde{s}=2 & \scriptscriptstyle \tilde{s}=3 \\
% \begin{block}{c(cccc)}
%   \scriptscriptstyle u^+=0 & \scriptstyle 0 & \scriptstyle 0 & \scriptstyle 0 & \scriptstyle 0 \\
%   \scriptscriptstyle u^+=1 & \scriptstyle 0 & \scriptstyle 1 & \scriptstyle 0 & \scriptstyle 1 \\
%   \scriptscriptstyle u^+=2 & \scriptstyle 1 & \scriptstyle 0 & \scriptstyle 1 & \scriptstyle 0 \\
%   \scriptscriptstyle u^+=3 & \scriptstyle 1 & \scriptstyle 1 & \scriptstyle 1 & \scriptstyle 1 \\
% \end{block} 
% \end{blockarray}

\begin{align}
\label{eq:fLA}
\scriptstyle x=f(u^+,\tilde{s}) =\scriptstyle
\begin{bmatrix}
%\rule{0pt}{12pt} &   & 0 & 1 & 2 & 3 \\
\scriptstyle 0 & \scriptstyle 0 & \scriptstyle 0 & \scriptstyle 0 \\
\scriptstyle 0 & \scriptstyle 1 & \scriptstyle 0 & \scriptstyle 1 \\
\scriptstyle 1 & \scriptstyle 0 & \scriptstyle 1 & \scriptstyle 0 \\
\scriptstyle 1 & \scriptstyle 1 & \scriptstyle 1 & \scriptstyle 1 \\
\end{bmatrix}
 \ , \
\scriptstyle P(u^+|u,Q=1)=
\begin{bmatrix}
\scriptstyle 0 & \scriptstyle \alpha & \scriptstyle \bar{\alpha} & \scriptstyle 0 \\
\scriptstyle 1 & \scriptstyle 0 & \scriptstyle 0 & \scriptstyle 0 \\
\scriptstyle 0 & \scriptstyle 0 & \scriptstyle 0 & \scriptstyle 1 \\
\scriptstyle 0 & \scriptstyle \beta & \scriptstyle \bar{\beta} & \scriptstyle 0 
\end{bmatrix} 
\ , \
\scriptstyle P(u^+|u,Q=2)=
\begin{bmatrix}
\scriptstyle 0 & \scriptstyle \beta & \scriptstyle \bar{\beta} & \scriptstyle 0  \\
\scriptstyle 0 & \scriptstyle 0 & \scriptstyle 0 & \scriptstyle 1 \\
\scriptstyle 1 & \scriptstyle 0 & \scriptstyle 0 & \scriptstyle 0 \\
\scriptstyle 0 & \scriptstyle \alpha & \scriptstyle \bar{\alpha} & \scriptstyle 0
\end{bmatrix}
\ , \
\; \alpha,\beta \in [0,1].
\end{align}
The transition matrix $P(\tilde{s}^+,u^+,q^+|\tilde{s},u,q)$ (given by Eq. \eqref{eq:suq_transition}) for the corresponding $(\tilde{S},U,Q)$-graph has a unique stationary distribution $\pi(\tilde{s},u,q)$ 
% with the following marginal distributions 
% \begin{align}
%     \pi_Q(1)= \pi_Q(2)
%     =\textstyle \frac{1}{2}, \quad \pi_{\tilde{S}}(\tilde{s})=\frac{1}{4}, \forall \tilde{s}\in \tilde{\cS},
% \end{align}
with the following conditional distribution matrices whose rows and columns represent $u$ and $\tilde{s}$, respectively:
\begin{align}
\scriptstyle \pi(u,\tilde{s}|Q=1)&=
\begin{bmatrix}
\scriptscriptstyle \frac{7}{64} & \scriptscriptstyle \frac{7}{64} & \scriptscriptstyle \frac{7}{64} & \scriptscriptstyle \frac{7}{64} \\
  \scriptscriptstyle \frac{21\alpha + 3\beta}{128} & \scriptscriptstyle \frac{21\alpha + 3\beta}{128} & \scriptscriptstyle \frac{7\alpha + \beta}{128} & \scriptscriptstyle \frac{7\alpha + \beta}{128}\\
\scriptscriptstyle \frac{8-7\alpha-\beta}{128} & \scriptscriptstyle \frac{8-7\alpha-\beta}{128} & \scriptscriptstyle \frac{3(8-7\alpha-\beta)}{128}  & \scriptscriptstyle \frac{3(8-7\alpha-\beta)}{128} \\
\scriptscriptstyle \frac{1}{64} & \scriptscriptstyle \frac{1}{64} & \scriptscriptstyle \frac{1}{64} & \scriptscriptstyle \frac{1}{64}
\end{bmatrix},
\scriptstyle \pi(u,\tilde{s}|Q=2)=
\begin{bmatrix}
\scriptscriptstyle \frac{1}{64} & \scriptscriptstyle \frac{1}{64} & \scriptscriptstyle \frac{1}{64} & \scriptscriptstyle \frac{1}{64}  \\
\scriptscriptstyle \frac{7\alpha + \beta}{128} & \scriptscriptstyle \frac{7\alpha + \beta}{128} & \scriptscriptstyle \frac{21\alpha + 3\beta}{128} & \scriptscriptstyle \frac{21\alpha + 3\beta}{128} \\
\scriptscriptstyle \frac{3(8-7\alpha-\beta)}{128} & \scriptscriptstyle \frac{3(8-7\alpha-\beta)}{128} & \scriptscriptstyle \frac{8-7\alpha-\beta}{128} & \scriptscriptstyle \frac{8-7\alpha-\beta}{128} \\
\scriptscriptstyle \frac{7}{64} & \scriptscriptstyle \frac{7}{64} & \scriptscriptstyle \frac{7}{64} & \scriptscriptstyle \frac{7}{64}
\end{bmatrix}.
\end{align}
%Further, one can calculate that for $Q=1$:
% \begin{align}
%   &\pi_{U,\tilde{S}|Q}(u,\tilde{s}|q)=\textstyle \frac{1}{2} \quad \forall u \in \cU, \tilde{s} \in \tilde{\cS}, q \in \cQ.
% \end{align}
Furthermore, by $P(u^+,u,y|q)= \sum_{\tilde{s}} \pi(u,\tilde{s}|q) P(u^+|u,q)P_{Y|X,\tilde{S}}(y|f(u^+,\tilde{s}),\tilde{s})$, one can calculate that
\begin{align}
I(U^+,U;Y|Q=1)&=I(U^+,U;Y|Q=2)=\textstyle H\left( \frac{44-21 \alpha -3 \beta}{64} \right)-\frac{1}{2} \left[ H\left( \frac{1}{8}\right) +1 \right]. \label{eq:RewardQ}
\end{align}
% induced from
% \begin{align}
% P(u^+,u,y|q)&= \sum_{s_{i-1}} \pi(u,s|q) P(u^+|u,q)P(y|f(u^+,s),s) \text{,}
% \end{align}
% yielding:
Recall that the BCJR-invariant property has to be satisfied. The BCJR-invariant property has $|\cS| \times |\cU| \times |\cY| \times |\cQ|=64$ constraints that can be reduced to the constraints:
\begin{align}
      &7 \alpha +\beta -4=0, \nn\\
      &49 \alpha ^2+14 \alpha  \beta -28 \alpha +\beta ^2-4 \beta =0, \nn\\
      &49 \alpha ^2+14 \alpha  \beta -84 \alpha +\beta ^2-12 \beta +32=0, \nn
\end{align}
which altogether have two solutions: $\textstyle \{\alpha=\frac{4}{7}, \beta=0\}$ or $\textstyle \{\alpha=\frac{3}{7}, \beta=1\}$. It can be verified that both solutions maximize \eqref{eq:RewardQ}. 
Substituting either of these solutions in \eqref{eq:RewardQ}, we conclude that  $C_{\text{FB-LA-}1} \ge \frac{1}{2} (1 - H(\frac{1}{8}))$, which completes the proof.
% \begin{align}
%     C_{\text{CSI-E}}&=\max_{p(u),x(u,s)}I(U;Y), \\ C_{\text{SI-E}}&=\max_{p(u|s),x(u,s)}(I(U;Y)-I(U;S)),
% \end{align}
% and $U$ is an auxiliary RV with a finite cardinality. 
\end{proof} 

\subsection{Unifilar FSCs}
A FSC is called a unifilar FSC if for any time $i$ the new channel state is a deterministic function of the current state, input and output, i.e., $s_i=g(s_{i-1},x_i,y_i)$ for some deterministic function $g(\cdot)$. The encoder can calculate all of the states causally from $s_0$, by using its sent inputs and the outputs feedback; this explains why it is a special case of the setting. Here, we pick three examples of strongly connected unifilar FSCs whose feedback capacity is known from the literature. For each example, we show analytically that its feedback capacity can be achieved by any of our two computation tools with $|\cU|=2$.
%In other words, on the one hand, we demonstrate how using the $Q$-graph method with a special $Q$-graph results in the explicit expression of the feedback capacity. On the other hand, the special graph of each problem is obtained from DP simulations.
\subsubsection{Trapdoor Channel}
\label{subsec:Trapdoor}
The trapdoor channel \cite{blackwell1961information}, is a unifilar FSC which has binary inputs, outputs and states, and its state evolution is $s^+ = s\oplus x\oplus y$. The channel output depends on the input and the channel state according to the $ZS$ channel topology (see Fig. \ref{fig:ZS-channel}).
%i.e., if the input and the channel state agree, then the output is equal to them. Otherwise, the output is a random instance. 
%the output is equal to the state with probability $p$, and to the input with probability $\bar{p}$, where $p\in[0,1]$ is the channel parameter.
%%%%%%%%%%%%%%%%%%%%%%%%%%%%%%%%%%%%%%%%%%%%%%%%%%%%%%%%%%%%%%%%%%%%%%%%%
% \begin{table}[b]
% \caption{The probability of $P(y_i|x_i,s_{i-1})$} \centering
% \label{table:prob_chem}
% \begin{tabular}[h]{||c|c|c|c||}
% \hline \hline
% $x$ & $s$ & $P(y=0|x,s)$ \\
% \hline \hline
% 0 & 0 & 1 \\
% \hline
% 0 & 1 & $1-p$ \\
% \hline
% 1 & 0 & $p$ \\
% \hline
% 1 & 1 & 0 \\
% \hline \hline
% \end{tabular}
% \end{table}
%%%%%%%%%%%%%%%%%%%%%%%%%%%%%%%%%%%
%that only the initial channel state is known to the encoder it follows that the encoder may know the states at all times. Hence, it is clear that the trapdoor channel under this assumption is equivalent to our setting where the state is known causally to the encoder because the encoder knows the the initial state in particular.
%\textcolor{blue}{It can be verified that the trapdoor channel is strongly connected since the encoder can drive the channel state to $S=i$, $i=1,2$ with as high probability as desired by transmitting a sufficiently large sequence of $X=i$.}
%The trapdoor channel has been studied in several works, but its general capacity for any channel parameter has remained an open problem. 
The feedback capacity of the trapdoor channel was shown in \cite{Permuter06_trapdoor_submit} to be $C_{\text{FB-Trapdoor}}=\log \phi$, where $\phi\triangleq \frac{\sqrt{5}+1}{2}$ is the known golden ratio.  
%In \cite{Permuter06_trapdoor_submit}, the setting was not assumed to have causal SI availability at the encoder, but it is still a special case of our general setting because of the unifilar FSC property which allows the encoder to compute the state sequence causally. 
%%%%%%%%%%%%%%%%%%%%%%%%%%%%%%%%%%%%%%%%%%%%%%
%They derived the capacity by formulating it as an average-reward DP and solving the corresponding Bellman equation.
Here, we provide an alternative achievability for the feedback capacity that can be derived using either of the tools presented, i.e., the DP formulation (Theorem \ref{theorem:MDP}) and the $Q$-graph (Theorem \ref{theorem:qgraph_LB}).
%%%%%%%%%%%%%%%%%%%%%%%%%%%%%%%%%%%%%%%%%%%%%%
%The feedback capacity of the unifilar FSC channel can be expressed with a multi-letter expression without auxiliary RVs \cite{Permuter06_trapdoor_submit}. Indeed, for this class of channels, there exists a simpler MDP formulation in that it exhibits a smaller DP state space. Our objective here is to demonstrate that the new achievable rate is optimal for two cases.
\begin{figure}[t]
\begin{center}
\begin{psfrags}
    \psfragscanon
    \psfrag{A}[\hspace{2cm}][][0.8]{$Q=1$}
    \psfrag{B}[\hspace{2cm}][][0.8]{$Q=2$}
    \psfrag{C}[\hspace{2cm}][][0.8]{$Q=3$}
    \psfrag{D}[\hspace{2cm}][][0.8]{$Q=4$}
    \psfrag{E}[\hspace{2cm}][][0.8]{$Y=1$\hspace{0.5cm}}
    \psfrag{F}[\hspace{2cm}][][0.8]{$Y=1$}
    \psfrag{G}[\hspace{2cm}][][0.8]{\raisebox{0.5cm}{$Y=0$}}
    \psfrag{H}[][][0.8]{\raisebox{2.5cm}{$Y=1$}\hspace{0.5cm}}
    \psfrag{I}[\hspace{2cm}][][0.8]{$Y=0$\hspace{0.5cm}}
    \psfrag{J}[\hspace{2cm}][][0.8]{\raisebox{2.5cm}{$Y=1$}\hspace{0.4cm}}
    \psfrag{K}[][][0.8]{\raisebox{0.5cm}{$Y=0$}}
    \psfrag{L}[\hspace{2cm}][][0.8]{\raisebox{0.2cm}{$Y=0$}}
\includegraphics[scale=0.4]{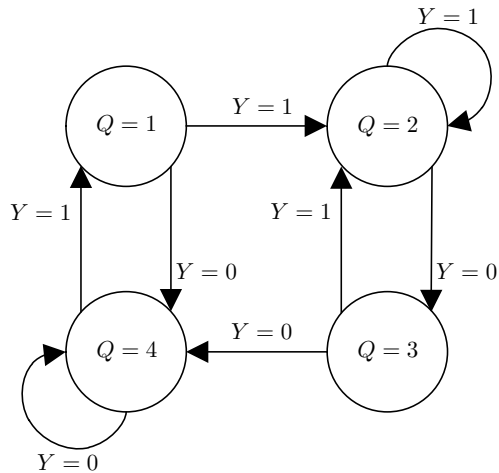}
\caption{A $Q$-graph for the trapdoor channel.} \label{fig:Qgraph_trapdoor}
\psfragscanoff
\end{psfrags}
\end{center}
\end{figure}
\begin{corollary}[by Theorem \ref{theorem:qgraph_LB}]
\label{corollary:trapdoor}
The $Q$-graph in Fig. \ref{fig:Qgraph_trapdoor} with $|\cU|=2$ achieves the known feedback capacity of the trapdoor channel, i.e.,
\begin{equation}
\label{eq:cap_trapdoor}
    C_{\text{FB-Trapdoor}}=
    \max_{\substack{\{P(u^+|u,q)\}\in\mathcal{P}_{\text{BCJR}} \\x=f(u^+,s)}}
I(U^+,U;Y|Q) =\log \phi,
\end{equation}
with specified $f(u^+,s)=u^+\oplus s$ and $P_{U^+|U,Q}$ given by $P_{U^+|U,Q}(0|0,q)=\frac{\sqrt{5}-1}{2}$ and $P_{U^+|U,Q}(0|1,q)=1$ for $q=1,...,4$.
\end{corollary}
The technical proof is given in Appendix \ref{appendix:ProofCorTrapdoor}.
The policy in Corollary \ref{corollary:trapdoor} can also be extracted from a standard evaluation of the VIA with the choice of
$|U|=2$ and $f(u^+,s)=u^+\oplus s$.
%\textcolor{blue}{We note that running the VIA for the DP formulation with this choice of $|\cU|$ and $f(\cdot)$ on a machine results in the same graph and lower bound as well.}
We also explain why this choice achieves the feedback capacity. Recall that the feedback capacity of connected unifilar FSCs \cite[Theorem~1]{Permuter06_trapdoor_submit} is
\begin{align}
    C_{\text{FB}}=\sup_{\{P(x_i|s_{i-1},y^{i-1})\}_{i\ge 1}} \liminf_{N \to \infty} \frac{1}{N} \sum_{i=1}^N I(X_i,S_{i-1};Y_i|Y^{i-1}). \label{eq:Cap_unifilar}
\end{align}
% where ${\{P(x_i|s_{i-1},y^{i-1})\}_{i\ge 1}}$ denotes the set of all distributions such that $P(x_i|s^{i-1},x^{i-1},y^{i-1})=P(x_i|s_{i-1},y^{i-1})$ for $i=1,2,\dots$.
Comparing our expression in \eqref{eq:supliminfcapCSI} with \eqref{eq:Cap_unifilar}
%the objective and maximization domain of \eqref{eq:supliminfcapCSI} to those of \eqref{eq:Cap_unifilar},
for any summand $i\ge 1$ gives:
\begin{align}
    I(U_i,U_{i-1};Y_i|Y^{i-1})&=H(Y_i|Y^{i-1})-H(Y_i|U_i,U_{i-1},Y^{i-1}) \nn\\
    &\stackrel{(a)}\le H(Y_i|Y^{i-1})-H(Y_i|X_i(U_i,S_{i-1}),S_{i-1},Y^{i-1}) \nn\\
    &=I(X_i(U_i,S_{i-1}),S_{i-1};Y_i|Y^{i-1}) \label{eq:uToUnifilar}
\end{align}
where (a) follows from the Markov chain $Y_i-(X_i,S_{i-1})-(U^{i},X^{i-1},Y^{i-1},S^{i-2})$. However, for our choice of $|U|=2$ and $f(u_i,s_{i-1})=u_i \oplus s_{i-1}$ it follows that $u_i=s_{i-1} \oplus x_i$, and from the channel model we then have $s_i=s_{i-1} \oplus x_i \oplus y_i=u_i \oplus y_i$ which implies that $s_{i-1}=u_{i-1} \oplus y_{i-1}$, i.e., $s_{i-1}$ is a deterministic function of $(u_{i-1},y_{i-1})$. Hence, (a) is achieved with equality since under this choice
\begin{align}
H(Y_i|U_i,U_{i-1},Y^{i-1})=H(Y_i|U_i,U_{i-1},Y^{i-1},S_{i-1}) =H(Y_i|X_i,S_{i-1},Y^{i-1}). \label{eq:secondHEquality}
\end{align}
% As for the maximization domain,
% \begin{align}
% P(u_i|u_{i-1},y^{i-1})=P(u_i|u_{i-1},y^{i-1},s_{i-1})=P(u_i\oplus s_{i-1}|u_{i-1},y^{i-1},s_{i-1}), \nn
% \end{align}
% which is a subdomain of $P(x_i|s_{i-1},y^{i-1})$ due to the fact that $x_i=u_i \oplus s_{i-1}$.

\subsubsection{Ising Channel}
The Ising channel is another unifilar FSC with $\cX=\cY=\cS=\{0,1\}$. This channel also has the $ZS$-channel topology as the trapdoor channel, given in Fig. \ref{fig:ZS-channel}, but the state evolution differs, i.e., $s^+=x$, therefore the feedback capacity differs as well.
%was invented as a problem in statistical mechanics by Lenz \cite{isingmodel1} and named after his student, Ernst Ising \cite{isingmodel2}.
%It was later introduced by Berger and Bonomi \cite{Berger90IsingChannel} as a problem in information theory of another unifilar FSC with binary inputs, outputs and states where $s^+ = x$.
%Hence, the Ising channel is obviously strongly connected since any state can be reached by transmitting its same symbol in $\cX$. The channel output of the Ising channel depends on the input and the channel state as given in Table \ref{table:prob_chem} with $p=0.5$, as well. 
Its feedback capacity was derived in \cite{elishco2014capacity,Ising_artyom_IT}, and was shown in \cite{elishco2014capacity} to be $C_{\text{FB-Ising}}=\max_{a\in [0,1]} 2 H(a)/(3+a) \approx 0.5755$. 
%\textcolor{blue}{Their setting was not assumed to have causal SI availability at the encoder, but it is still a special case of our general setting due to the unifilar FSC property.}
%%%%%%%%%%%%%%%%%%%%%%%%%%%%%%%%%%%
%that only the initial channel state is known to the encoder it follows that the encoder may know the states at all times. Hence, it is clear that the trapdoor channel under this assumption is equivalent to our setting where the state is known causally to the encoder because the encoder knows the the initial state in particular.
%The feedback capacity of the unifilar FSC channel can be expressed with a multi-letter expression without auxiliary RVs \cite{Permuter06_trapdoor_submit}. Indeed, for this class of channels, there exists a simpler MDP formulation in that it exhibits a smaller DP state space. Our objective here is to demonstrate that the new achievable rate is optimal for two cases.
%The feedback capacity of the unifilar FSC channel can be expressed with a multi-letter expression without auxiliary RVs \cite{Permuter06_trapdoor_submit}. Indeed, for this class of channels, there exists a simpler MDP formulation in that it exhibits a smaller DP state space. Our objective here is to demonstrate that the new achievable rate is optimal for two cases.
\begin{corollary}[by Theorem \ref{theorem:qgraph_LB}] \label{corollary:ising}
The $Q$-graph in Fig. \ref{fig:Qgraph_ising} with $|\cU|=2$ achieves the known feedback capacity of the Ising channel, i.e., 
\begin{equation}
\label{eq:cap_ising}
    C_{\text{FB-Ising}}= \max_{a \in[0,1]} \frac{2 H(a)}{3+a}.
\end{equation}
with specified $f(u^+,s)=u^+$ and $P_{U^+|U,Q}$ given by $P_{U^+|U,Q}(0|0,i)=1$, $P_{U^+|U,Q}(0|1,q)=0$ for $q=1,3$, and $P_{U^+|U,Q}(0|0,2)=a$,
$P_{U^+|U,Q}(0|1,4)=\bar{a}$ for some $a \in [0,1]$. 
\end{corollary}
The technical proof is given in Appendix \ref{appendix:ProofCorIsing}, where it is also clarified that $P_{U^+|U,Q}(0|1,2)$ and $P_{U^+|U,Q}(0|0,4)$ are irrelevant due to the fact that $P_{U|Q}(1|2)=0$, $P_{U|Q}(0|4)=0$.
\begin{figure}[t]
\begin{center}
\begin{psfrags}
    \psfragscanon
    \psfrag{A}[\hspace{2cm}][][0.8]{$Q=1$}
    \psfrag{B}[\hspace{2cm}][][0.8]{$Q=2$}
    \psfrag{C}[\hspace{2cm}][][0.8]{$Q=3$}
    \psfrag{D}[\hspace{2cm}][][0.8]{$Q=4$}
    \psfrag{E}[\hspace{2cm}][][0.8]{$Y=0$}
    \psfrag{F}[\hspace{2cm}][][0.8]{$Y=1$}
    \psfrag{G}[\hspace{2cm}][][0.8]{$Y=0$\hspace{-0.1cm}}
    \psfrag{H}[\hspace{2cm}][][0.8]{$Y=0$\hspace{0.1cm}}
    \psfrag{I}[\hspace{2cm}][][0.8]{$Y=1$}
    \psfrag{J}[\hspace{2cm}][][0.8]{$Y=1$\hspace{0.1cm}}
    \psfrag{K}[][][0.8]{$Y=1$\hspace{-0.1cm}}
    \psfrag{L}[\hspace{2cm}][][0.8]{$Y=0$}
\includegraphics[scale=0.40]{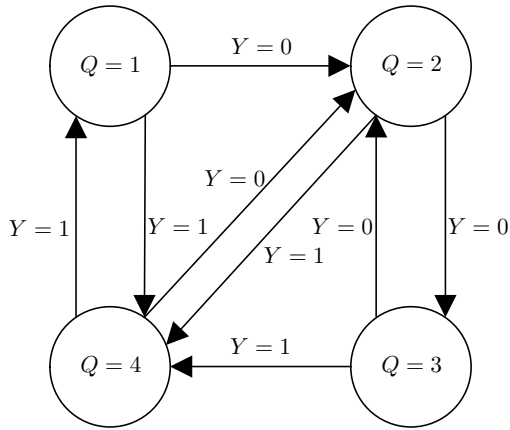}
\caption{A $Q$-graph for the Ising channel.} \label{fig:Qgraph_ising}
\psfragscanoff
\end{psfrags}
\end{center}
\end{figure}
%%%%%%%%%%%%%%%%%%%%%%%%%%%%%%%%%%%%%%%%%%%%%%%%
% We explain why the capacity is achieved by comparing the objective and maximization domain of \eqref{eq:supliminfcapCSI} to that of \eqref{eq:Cap_unifilar}. According to the channel model, the choice of $|U|=2$ and $x=u^+$ implies that $s_{i}=x_{i
% }=u_{i}$,  i.e., $s_{i-1}=u_{i-1}$. Consequently, (a) in \eqref{eq:uToUnifilar} is achieved with equality since \eqref{eq:secondHEquality} holds. Clearly, the maximization domain is $P(u_i|u_{i-1},y^{i-1})=P(x_i|s_{i-1},y^{i-1})$. 
%%%%%%%%%%%%%%%%%%%%%%%%%%%%%%%%%%%%%%%%%%%%%%%%
%%Let $C=[z_1 z_2 z_3 z_4]$ represent a quadruplet of DP states, where:
%\begin{align}
%z1 \triangleq\ P(s_{t-1}=0,u_{t-1}=0|y^{t-1}), \quad
%z2 \triangleq\ P(s_{t-1}=0,u_{t-1}=1|y^{t-1}), \nonumber\\
%z3 \triangleq\ P(s_{t-1}=1,u_{t-1}=0|y^{t-1}), \quad
%z4 \triangleq\ P(s_{t-1}=1,u_{t-1}=1|y^{t-1}).\nonumber
%\end{align}

\subsubsection{Input-constrained BEC}
\label{subsec:constrainedBEC}
This channel has a binary input sequence with
the $(1,\infty)$-RLL constraint, i.e., it contains no consecutive ones; and instead of binary output \cite{Sabag_BIBO_IT} it comprises of a BEC. The input-constrained BEC, which does not fall into the classical definition of unifilar FSCs, can be considered so by viewing the input letter as a channel state representing the input constraint, i.e., $s^+=x$. The feedback capacity of this problem was calculated in \cite{Sabag_BEC} 
to be $\max_{p \in [0,0.5]} \frac{H(p)}{\frac{1}{\bar{\epsilon}}+p}$, where $\epsilon$ is the erasure probability parameter, and it was generalized in \cite{PeledSabagBEC} for the case of the $(0, k)$-RLL when swapping ‘$0$’s and ‘$1$’s.
%The $(1,\infty)$-RLL BEC is %a memoryless channel, but its inputs are constrained to the $(1,\infty)$-RLL constraint, i.e., no consecutive ones are allowed. This constraint introduces a memory to the setting, but it can be formulated as an equivalent unifilar FSC. 
%%%%%%%%%%%%%%%%%
%The feedback capacity of the unifilar FSC channel can be expressed with a multi-letter expression without auxiliary RVs \cite{Permuter06_trapdoor_submit}. Indeed, for this class of channels, there exists a simpler MDP formulation in that it exhibits a smaller DP state space. Our objective here is to demonstrate that the new achievable rate is optimal for two cases.
\begin{corollary}[by Theorem \ref{theorem:qgraph_LB}]
\label{corollary:BEC}
The $Q$-graph in Fig. \ref{fig:Qgraph_bec} with $|\cU|=2$ achieves the known feedback capacity of the 
input-constrained BEC, i.e.,
\begin{equation}
\label{eq:cap_bec}
   C_{\text{FB-cBEC}}= \max_{p \in [0,0.5]} \frac{H(p)}{\frac{1}{\bar{\epsilon}}+p}.
\end{equation}
with specified
\begin{align}
\label{eq:x_bec}
    f(u^+,s)=\begin{cases}
            1, & s=0 \text{ and } u^+=1\\
            0, & \text{otherwise}
		   \end{cases}
\end{align}
and $P_{U^+|U,Q}$ given by $P_{U^+|U,Q}(0|1,1)=1$, $P_{U^+|U,Q}(0|0,2)=\bar{p}$, $P_{U^+|U,Q}(0|0,3)=\frac{1-2p}{\bar{p}}$, $P_{U^+|U,Q}(0|1,3)=1$, for some $p \in [0,0.5]$. 
\end{corollary}
The technical proof is given in Appendix \ref{appendix:ProofCorBEC}, where it is also clarified that $P_{U^+|U,Q}(0|0,1)$ and $P_{U^+|U,Q}(0|1,2)$ are irrelevant due to the fact that $P_{U|Q}(0|1)=0$, $P_{U|Q}(1|2)=0$.
\begin{figure}[t]
\begin{center}
\begin{psfrags}
    \psfragscanon
     \psfrag{A}[\hspace{2cm}][][0.8]{$Q=1$}
     \psfrag{B}[\hspace{2cm}][][0.8]{$Q=2$}
     \psfrag{C}[\hspace{2cm}][][0.8]{$Q=3$}
     \psfrag{D}[\hspace{2cm}][][0.8]{$Y=0 \quad Y=?$}
     \psfrag{E}[\hspace{2cm}][][0.8]{$Y=1$}
     \psfrag{F}[\hspace{2cm}][][0.8]{$Y=0$}
     \psfrag{G}[\hspace{2cm}][][0.8]{$Y=?$}
     \psfrag{H}[\hspace{2cm}][][0.8]{$\quad Y=0$}
     \psfrag{I}[\hspace{2cm}][][0.8]{$Y=?$}
     \psfrag{J}[\hspace{2cm}][][0.8]{$Y=1$}
    %\psfrag{K}[\hspace{2cm}][][1]{$Y=1$}
    %\psfrag{L}[\hspace{2cm}][][1]{$Y=0$}
\includegraphics[scale=0.4]{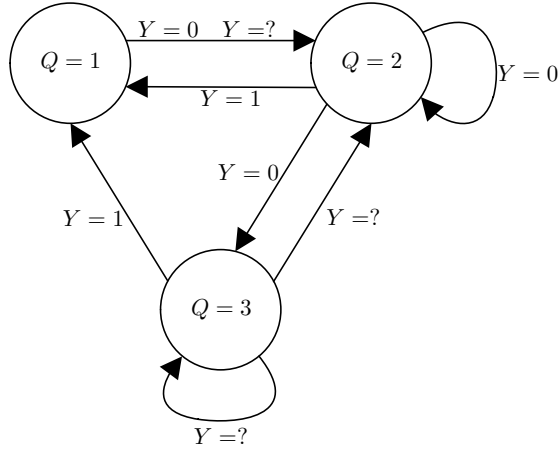}
\caption{A $Q$-graph for the input constrained BEC.} \label{fig:Qgraph_bec}
\psfragscanoff
\end{psfrags}
\end{center}
\end{figure}
Next, we focus on a generalization the Ising channel with a stochastic state evolution rather than deterministic.

\subsection{The Noisy-Ising($\eta$) Channel}
\label{subsec:noisyIsing}
We study a generalization of the Ising channel, where the channel state is obtained as the output of a binary symmetric channel (BSC) whose input is the channel input. That is, $P(s_i|x_i)$ is a BSC with crossover probability $\eta\in [0,1]$; we call this generalized channel the \textit{noisy-Ising($\eta$) channel} and denote its feedback capacity given a parameter $\eta$ by $C_{\text{fb-csi}}^{\text{n-Ising}}(\eta)$. The problem symmetry implies that its feedback capacity is symmetric in $\eta$, i.e., $C_{\text{fb-csi}}^{\text{n-Ising}}(\bar{\eta})=C_{\text{fb-csi}}^{\text{n-Ising}}(\eta)$, thus our focus is limited to $\eta \in [0,0.5]$. For these values of $\eta$, we provide an analytic, closed-form lower bound in Theorem \ref{theorem:LB_noisyIsing} below and two numerical lower bounds using the DP and $Q$-graph methods. 
% as illustrated in Fig. \ref{fig:capNoisyIsing}.
% to obtain numerical achievable rates for $\eta\in [0,0.5]$, denoted by $R_{\text{DP}}$, Using the $Q$-graph method
\begin{figure}[t]
\begin{center}
    \includegraphics[scale=0.65]{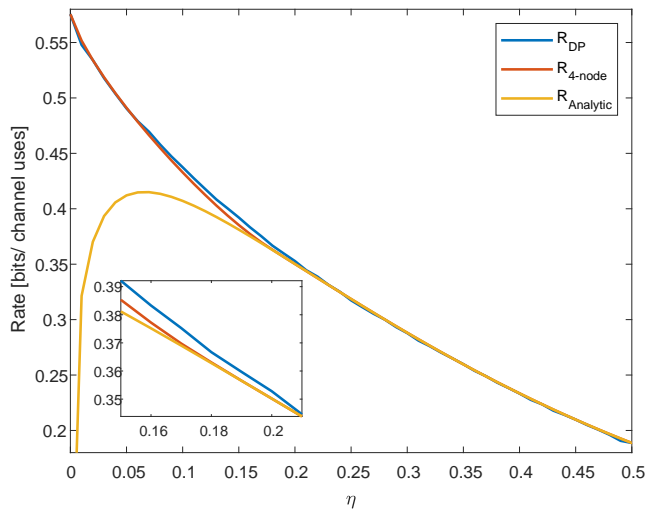}
\end{center}
\caption{Lower bounds on the feedback capacity of the noisy-Ising($\eta$) channel:
$R_{\text{Analytic}}$ are analytically given in \eqref{eq:LB-N-Ising}, $R_{\text{DP}}$ and $R_{4\text{-node}}$ are calculated numerically based on the DP tool and the $4$-node $Q$-graph in Fig. \ref{fig:Qgraph_ising}, respectively.}
\label{fig:capNoisyIsing}
\end{figure}

\begin{theorem}
\label{theorem:LB_noisyIsing}
The feedback capacity of the noisy-Ising($\eta$) channel with SI known causally at the encoder is lower-bounded by
\begin{align}
   C_{\text{fb-csi}}^{\text{n-Ising}}(\eta)& \ge \textstyle H\left(\frac{(2-\eta)(a \bar{\eta}+\eta)}{a-2a\eta+2}\right)-\frac{2-a\eta-a}{a-2a\eta+2}H\left(\frac{\eta}{2}\right)-\frac{a(2-\eta)}{a-2a\eta+2}H\left(\frac{\bar{\eta}}{2}\right), \quad \eta\in[0,0.5], \label{eq:LB-N-Ising}
\end{align}
where $a=\frac{4 \eta ^2+\sqrt{16 \eta ^2-31 \eta +16} \sqrt{\eta }-5 \eta }{4 \eta ^2-10 \eta +4}$ for any $\eta\in [0,0.5)$, and $a=\frac{1}{3}$ if $\eta=0.5$.
% \begin{align}
% a= 
% \begin{cases}
% \frac{4 \eta ^2+\sqrt{16 \eta ^2-31 \eta +16} \sqrt{\eta }-5 \eta }{4 \eta ^2-10 \eta +4}, & \eta\in [0,0.5)\\
% \frac{1}{3}, & \eta=0.5
% \end{cases}.
% \label{eq:aExplicit}
% \end{align}
\end{theorem}
Theorem \ref{theorem:LB_noisyIsing} is a direct consequence of Theorem \ref{theorem:qgraph_LB} with the $2$-node $Q$-graph in Fig. \ref{fig:Qgraph_2node} with $|\cU|=2$, function $x=u^+$ and policy
\begin{align}
 P_{U^+|U,Q}(1|0,1)&=P_{U^+|U,Q}(0|1,2)=a\in[0,1], \nn\\
 P_{U^+|U,Q}(1|1,1)&=P_{U^+|U,Q}(0|0,2)=1,
 \label{eq:policyR_analytic}
\end{align}
and the resulting BCJR constraints reduce to the quadratic equation in $a$:
\begin{align}
    (2 - 5 \eta + 2 \eta^2) a^2 +(5 - 4 \eta) \eta a -2 (1 - \eta) \eta=0, \quad \eta\in[0,0.5], \label{eq:aEquation}
\end{align}
whose positive solution is $a$ given in the theorem.
Specifically, for $\eta=0.5$, an i.i.d. state is obtained, and 
\begin{align}
R_{\text{Analytic}}(0.5)=\textstyle H\left(\frac{3(a+1)}{8}\right)-H\left(\frac{1}{4}\right)=\textstyle 1-H\left(\frac{1}{4}\right) \approx 0.1887,
\end{align}
where $a=\frac{1}{3}$ maximizes \eqref{eq:LB-N-Ising} and satifies the BCJR condition \eqref{eq:aEquation} simultaneously. This rate is the capacity itself, which is
$C_{\text{CSI-E}}$ previously given in Table \ref{table:BoundsLA}. The lower bound in Theorem \ref{theorem:LB_noisyIsing} for any $\eta \in [0,0.5]$ is illustrated in Fig. \ref{fig:capNoisyIsing}, and denoted by $R_{\text{Analytic}}$.
%as illustrated in Fig. \ref{fig:capNoisyIsing}.
% to obtain numerical achievable rates for $\eta\in [0,0.5]$, denoted by $R_{\text{DP}}$, Using the $Q$-graph method
% $R_{\text{Analytic}}$ as follows. The optimal identified $2$-node $Q$-graph is the one previously given in Fig. \ref{fig:Qgraph_2node}, and by taking 
% \begin{align}
%  P_{U^+|U,Q}(1|0,1)&=P_{U^+|U,Q}(0|1,2)=a\in[0,1], \nn\\
%  P_{U^+|U,Q}(1|1,1)&=P_{U^+|U,Q}(0|0,2)=1,
%  \label{eq:policyR_analytic}
% \end{align}
% \begin{align}
%     \text{UB}_{\text{fb-csi}}^{\text{n-Ising}}(\eta)&= \textstyle\max_{a\in [0,1]} H\left(\frac{(2-\eta)(a \bar{\eta}+\eta)}{a-2a\eta+2}\right)-\frac{2-a\eta-a}{a-2a\eta+2}H\left(\frac{\eta}{2}\right)-\frac{a(2-\eta)}{a-2a\eta+2}H\left(\frac{\bar{\eta}}{2}\right), \quad \eta\in[0,0.5], \label{eq:UB-N-Ising}
% \end{align}
%Using the same $2$-node $Q$-graph we used to derive this upper bound, we derive a closed-form, simple formula for an achievable rate as follows:

\begin{remark}
The choice of $P_{U^+|U,Q}$ in \eqref{eq:policyR_analytic}, not only provides a simple, closed-form achievable rate, but it also implies a simple coding scheme that utilizes the choice of $x=u^+$. Using the feedback, the encoder can compare $y_{i-1}$ with the transmitted symbol $u_{i-1}=x_{i-1}$. If $y_{i-1}\ne x_{i-1}$, i.e., the decoder received the wrong symbol, then the encoder keeps transmitting the same symbol until $y_{i-1}=x_{i-1}$, i.e., until the decoder received the desired symbol. Only after the decoder receives it properly, the encoder transmits the opposite symbol with probability $a$.
\end{remark}

We implemented the $Q$-graph method (Theorem \ref{theorem:qgraph_LB}) on this channel with $|\cU|=2,3,\dots$ and noticed that increasing this cardinality beyond $2$ does not improve the $Q$-graph lower bound. Thus, we suspect that for this channel, the auxiliary RVs in \eqref{eq:capTrunc} have cardinality $|\cU|=2$. Fixing $|\cU|=2$, we compared the resulting achievable rates of all the strategy functions $f:\cU \times \cS \to \cX$ via the DP and $Q$-graph methods, and discovered empirically that the optimal function with the greatest rates is, in fact, $x=u^+$, for any arbitrary $\eta \in [0,0.5]$ that was examined. 

After discovering the optimality of the $x=u^+$ function assuming $|\cU|=2$, we utilized the DP method (by the VIA) and the $Q$-graph method with $|\cQ|=4$ to evaluate numerical achievable rates for $\eta\in [0,0.5]$, denoted by $R_{\text{DP}}$ and $R_{4-\text{node}}$, respectively, as given in Fig. \ref{fig:capNoisyIsing}. For the latter, we identified that the optimal $4$-node $Q$-graph is the same $Q$-graph for the original Ising channel previously shown in Fig. \ref{fig:Qgraph_ising}. It can be shown that $R_{\text{DP}}$ and $R_{4-\text{node}}$ are approximate for all $\eta\in [0,0.5]$, and $R_{\text{Analytic}}$ is approximate to both of them for $\eta \in [0.15,0.5]$.

We note that if the choice $|\cU|=2, x_i=u_i$ in \eqref{eq:capTrunc}
is indeed optimal, that is, the cardinality of $|\cU|$ is bounded and the optimal strategy function is independent of the state, it implies that the SI known causally at the encoder does not increase the feedback capacity in this example. On the one hand, substituting $u_i$ by $x_i$ in the capacity expression \eqref{eq:capTrunc} gives the achievable rate
\begin{align}
    R_{\text{fb-csi}}= \lim_{N \to \infty} \frac{1}{N}  \max_{\substack{\{P(x_i|x_{i-1},y^{i-1})\}_{i=1}^N}} \sum_{i=1}^{N} I(X_i,X_{i-1};Y_i|Y^{i-1}),  \label{eq:LB_Psx}
\end{align}
which is the feedback capacity of the noisy-Ising($\eta$) channel with SI known causally at the encoder (if the mentioned choice is optimal). On the other hand, it can be shown that \eqref{eq:LB_Psx} characterizes the feedback capacity of any connected FSC without SI where the state depends only on the input, $P_{S^+|X}$, as we can define a new state $x_{i-1}$. That is, the Markov chain $P(y_i|x^i,y^{i-1})=P(y_i|x_i,x_{i-1})= \sum_{s_{i-1}}P_{S^+|X}(s_{i-1}|x_{i-1})Q_{Y|X,S}(y_i|x_i,s_{i-1})$ holds, and the DI between $X^N$ to $Y^N$ reduces to \eqref{eq:LB_Psx}.
%%%%%%%%%%%%%%%%%%%%%%%%%%%
% The following is our conjecture regarding the influence of the availability of CSI known at the encoder on the feedback capacity of the noisy-Ising($\eta$) channel.
% \begin{guess}
% \label{conjecture:noisyIsing}
% CSI at the encoder does not increase the feedback capacity of the noisy-Ising($\eta$) channel.
% \end{guess}
%%%%%%%%%%%%%%%%%%%%%%%%%%%%%
%Although we do not prove Conjecture \ref{conjecture:noisyIsing} analytically, we establish its explanation as follows****

\begin{remark}
For each of the given examples, the corresponding choice of $|\cU|$ and $f(\cdot)$ satisfies the Markov chain 
\begin{align}
\label{markovCard}
S_{i-1}-(U_{i-1},Y^{i-1})-U^{i-2},    
\end{align}
which in turn satisfies the Markov chain $Y_i-(U_i,U_{i-1},Y^{i-1})-U^{i-2}$ and renders the DI between $U^N$ and $Y^N$ in the first capacity expression \eqref{eq:capDI} the alternative capacity expression \eqref{eq:capTrunc} (which generally lacks a cardinality bound $|\cU|$), i.e., $I(U^N \to Y^N)=\sum_{i=1}^N I(U_i,U_{i-1};Y_i|Y^{i-1})$. The Markov chain is shown as follows:
\begin{align}
P(y_i|u^{i},y^{i-1})&\stackrel{(a)}=\sum_{s_{i-1}} P(s_{i-1}|u^{i-1},y^{i-1})P(y_i|u^{i},s_{i-1},y^{i-1}) \nn \\
&\stackrel{(b)}=\sum_{s_{i-1}} P(s_{i-1}|u_{i-1},y^{i-1})P(y_i|u^{i},s_{i-1},y^{i-1}) \nn \\
&\stackrel{(c)}=\sum_{s_{i-1}} P(s_i|u_{i-1},y^{i-1})P_{Y|X,S}(y_i|f(u_i,s_{i-1}),s_{i-1}) \nn\\
&=P(y_i|u_i,u_{i-1},y^{i-1}),
\end{align}
where
\begin{enumerate}[label={(\alph*)}]
\item follows from the law of total probability and the Markov chain $U_i-(U^{i-1},Y^{i-1})-S_{i-1}$ implied from the joint distribution in \eqref{eq:joint_distDI};
\item follows from the choice of $|\cU|$ and $f(\cdot)$ inducing
the Markov chain \eqref{markovCard};
\item follows from the channel mode.
\end{enumerate}
%We suspect that generally, for  such a choice of $|\cU|,f(\cdot)$ that satisfies the Markov chain in \eqref{markovCard} also
%attributes an upper bound on the cardinality $|\cU|$ for the capacity expression.
\end{remark}

\section{Proof of the Feedback Capacity}
\label{sec:proof_cap}
In this section, we prove Theorem \ref{theorem:capDI}. In the achievability part (Sec. \ref{subsec:achievabilityProof}), we show that any rate less than 
$\lim_{N \to \infty} \frac{1}{N}  \max_{P(u^N||y^{N-1},s_0)} I(U^N \to Y^N)$ is achievable, while in the converse part (Sec. \ref{subsec:converseProof}), we show that rates greater than $\lim_{N \to \infty} \frac{1}{N}  \max_{P(u^N||y^{N-1})} I(U^N \to Y^N)$ are not achievable. Because the former expression is greater than the latter due to the maximization domain that may depend on $s_0$, it will be deduced that \eqref{eq:capDI} characterizes the capacity of the setting. The proof will be concluded by showing the cardinality bound $|\cU|\le|\cX|^{|\cS|}$ as given in Appendix \ref{appendix:cardinality}. In the last part (Sec.~\ref{subsec:computableN}) we explain the derivation of Theorem \ref{theorem:anyN_LB_UB} based on the achievability and converse proofs.
%%%%%%%%%%%%%%%%%%%%%%%%%%%%%%%%%%%%%%%%%%%%%
\subsection{Theorem \ref{theorem:capDI} - Proof of Achievability}
\label{subsec:achievabilityProof}
We prove that every rate $R< \lim_{N \to \infty} \frac{1}{N}   \max_{\substack{P(u^N||y^{N-1},s_0)}} I(U^N\rightarrow Y^N)$, where $\cU$ is the set of all strategies, is achievable. The proof 
is established on the feedback capacity expression of FSCs without SI, and comprises of three main steps shown in \eqref{eq:Ach1}-\eqref{eq:Ach3} below.

The feedback capacity of any FSC without SI was shown in \cite{PermuterWeissmanGoldsmith09} to be lower bounded by
\begin{align}
\label{eq:LB_anyFSC}
    C_{\text{fb}}&\ge\lim_{N \to \infty} \frac{1}{N}  {\max_{\substack{P(x^N||y^{N-1})}}} \min_{s_0} I(X^N\rightarrow Y^N|s_0).
\end{align}
Assuming that for such a FSC the encoder is informed of the initial state at the beginning of each communication block, it immediately follows that the feedback capacity is lower bounded by
\begin{align}
    C_{\text{fb}}&\ge\lim_{N \to \infty} \frac{1}{N}  {\max_{\substack{P(x^N||y^{N-1},s_0)}}} \min_{s_0} I(X^N\rightarrow Y^N|s_0). \label{eq:LB_NOcsi}
\end{align}
Based on \eqref{eq:LB_NOcsi}, we conclude that the feedback capacity of the setting, i.e., strongly connected FSCs with SI known causally at the encoder, is lower bounded by
\begin{align}
    C_{\text{fb-csi}}&\ge\lim_{N \to \infty} \frac{1}{N}  {\max_{\substack{P(u^N||y^{N-1},s_0)}}} \min_{s_0} I(U^N\rightarrow Y^N|s_0) \label{eq:Ach1} \\
    &=\lim_{N \to \infty} \frac{1}{N}  {\max_{\substack{P(u^N||y^{N-1},s_0)}}} I(U^N\rightarrow Y^N|S_0) \label{eq:Ach2} \\
    &=\lim_{N \to \infty} \frac{1}{N}  {\max_{\substack{P(u^N||y^{N-1},s_0)}}} I(U^N\rightarrow Y^N), \label{eq:Ach3}
\end{align}
where $\cU$ is the set of all strategies, and for \eqref{eq:Ach2} $P(s_0)$ is arbitrary.
\begin{figure}[t]
\begin{center}
\begin{psfrags}
    \psfragscanon
    \psfrag{E}[][][0.7]{$M$}
    \psfrag{S}[][][0.7]{\begin{tabular}{@{}l@{}}
    $S^{i-1}$
  %  \hspace{2cm}
  % State with memory\\
  % \hspace{2cm}
  % e.g. $S^{i-1}$
    \end{tabular}}
    %{\\}
    \psfrag{A}[\hspace{2cm}][][0.75]{Encoder}
	 \psfrag{F}[\hspace{2cm}][][0.75]{$U_i(M,Y^{i-1})$}
	 \psfrag{W}[\hspace{2cm}][][0.75]{$P(s_i,y_i|x_i,s_{i-1})$}
	 \psfrag{M}[\hspace{2cm}][][0.75]{Finite State Channel}
	 \psfrag{P}[\hspace{2cm}][][0.75]{\textbf{New FSC  $P(s_i,y_i|u_i,s_{i-1})$}}
	 \psfrag{G}[][][0.75]{$Y_i$}
	 \psfrag{C}[\hspace{2cm}][][0.75]{Decoder}
	 \psfrag{R}[\hspace{2cm}][][0.75]{$S_{i-1}$}
	 \psfrag{K}[][][0.75]{$\hat{M}$}
	 \psfrag{H}[\hspace{2cm}][][0.75]{$Y_i$}
	 \psfrag{D}[\hspace{2cm}][][0.75]{Delay}
	 \psfrag{Z}[\hspace{2cm}][][0.75]{$X_i$}
	 \psfrag{J}[\hspace{2cm}][][0.75]{$Y_{i-1}$}
	 \psfrag{L}[\hspace{2cm}][][0.75]{Finite-State Channel}
	 \psfrag{T}[\hspace{2cm}][][0.75]
	 {$x_i=f(u_i,s_{i-1})$}
%\psfrag{tag}[][][<scale>]{Latex Text}
\includegraphics[scale=0.7]{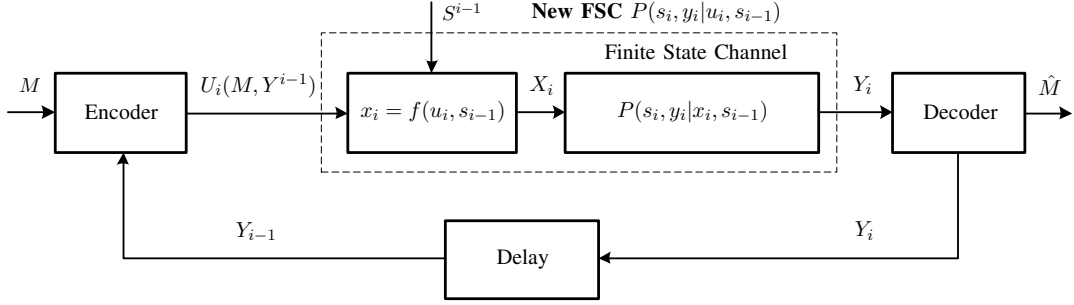}
\caption{An equivalent setting for the lower bound formulation with a new FSC $P_{Y,S^+|U,S}$ where the state is not available at the encoder.} \label{fig:shannon_strategy}
\psfragscanoff
\end{psfrags}
\end{center}
\end{figure}

\textit{Proof of Inequality \eqref{eq:Ach1}:} 
We introduce a new FSC $P_{Y,S^+|U,S}$ with input $U$ (instead of input $X$) without SI (see Fig. \ref{fig:shannon_strategy}), where the encoder is informed of the initial state at the beginning of each communication block. 
Inducing a new FSC without SI can be done as follows. Instead of directly encoding over the input alphabet $\cX$, at time $i$, the transmitter encodes an auxiliary RV $u_i(m,y^{i-1})$ of alphabet $\cU$, and transmits $x_i=f(u_i,s_{i-1})$, where  $f:\cU \times \cS \to \cX$ is a time-invariant function. The induced FSC is $P_{Y,S^+|U,S}=P_{Y,S^+|X=f(U,S),S}$, with 
input $U_i$, current state $S_{i-1}$, output $Y_i$ and new state $S_i$. Thus, using \eqref{eq:LB_NOcsi}, we can replace $X_i$ with $U_i$ in this expression in order to lower bound $C_{\text{fb-csi}}$ as in \eqref{eq:Ach1}.

\textit{Proof of Equality \eqref{eq:Ach2}:} The proof of this step follows in the same manner as of the proof of \cite[Equality (15) and (16)]{Permuter06_trapdoor_submit} (Steps (a)-(d) there) and considering the following modifications. Replace $X_i$ with $U_i$, and in Step (c) use the following lemma instead of \cite[Lemma 2]{Permuter06_trapdoor_submit}. 
\begin{lemma}
\label{lemma:LemmaConnected}
(Analogue to \cite[Lemma 2]{Permuter06_trapdoor_submit}) 
For a connected FSC with SI known causally at the encoder, given any input distribution $P_1(u^N||y^{N-1},s_0)$ and any $s'_0$, there exists an input distribution $P_2(u^N||y^{N-1},s'_0)$ such that
\begin{align}
  \frac{1}{N} | I_{P_1}(U^N\rightarrow Y^N|s_0)-I_{P_2}(U^N\rightarrow Y^N|s'_0)| \le \frac{c}{N},    
\end{align}
where $c$ is a constant that does not depend on $N,s_0,s'_0$. The term $I_{P_1}(U^N\rightarrow Y^N|s_0)$ denotes the DI induced by $P_1(u^N||y^{N-1},s_0)$, where $s_0$ is the initial state. Similarly, $I_{P_2}(U^N\rightarrow Y^N|s_0)$ denotes the DI induced by $P_2(u^N||y^{N-1},s'_0)$, where $s'_0$ is the initial state.
\end{lemma}
The proof of Lemma \ref{lemma:LemmaConnected} is given in Appendix \ref{appendix:proof_LemmaConnected}.

\textit{Proof of Equality \eqref{eq:Ach3}:} The last step follows from \cite[Lemma 4]{PermuterWeissmanGoldsmith09}), i.e., for any joint distribution $P(u^n,y^n,s_0)$, $|I(U^N\to Y^N)-I(U^N \to Y^N|S_0)|\le \log |\cS|$.
%$U^N,Y^N,S_0$ be arbitrary random vectors and $S\in \cS$ a RV, where  
%\begin{align}
%\end{align}
Finally, we claim that the limit in \eqref{eq:Ach1} exists as follows. Define $\underline{\tilde{C}}_N \triangleq \frac{1}{N} \max_{P(u^N||y^{N-1},s_0)} \min_{s_0} I(U^N \to Y^N|s_0)$. $\lim \underline{\tilde{C}}_N$ exists due to the super-additive property of the sequence $N(\underline{\tilde{C}}_N-\frac{\log{\cS}}{N})$, which follows from the proof of \cite[Th. 4]{PermuterWeissmanGoldsmith09} when replacing $X_i$ with $U_i$ and conditioning on $s_0$ in the CCD.
\qed
%%*****************
\subsection{Theorem \ref{theorem:capDI} - Proof of Converse}
\label{subsec:converseProof}
Here, we prove that an achievable rate $R$ of any FSC (not necessarily connected) with feedback and SI known casually at the encoder must satisfy $R \le \lim_{N\to\infty} \frac{1}{N} \max_{P(u^N||y^{N-1}),{\{x_i=f(u_i,s_{i-1})\}}_{i=1}^N}
 I(U^N \to Y^N)$.
\begin{proof}
For a fixed sequence of $(2^{nR},n)$ codes such that $P_e^{(n)}\to 0$ as $n \to \infty$, we bound their achievable rate as
\begin{align}
    nR-n\epsilon_n&\stackrel{(a)}\leq I(M;Y^n) \nonumber\\
    &=\sum_{i=1}^n I(M;Y_i|Y^{i-1}) \nonumber\\
    &\stackrel{(b)}= \sum_{i=1}^n I(U_i;Y_i|Y^{i-1}) \nonumber\\
    &\stackrel{(c)}\leq 
    \max_{{\{P(u_i|u_{i-1},y_{i-1}),P(x_i|u_i,s_{i-1})\}}_{i=1}^n} \sum_{i=1}^n I(U_i;Y_i|Y^{i-1}) \nonumber\\
       &\stackrel{(d)}= 
    \max_{{\{P(u_i|u_{i-1},y^{i-1}),P(v_i),x_i=f_i(u_i,v_i,s_{i-1})
        \}}_{i=1}^n} \sum_{i=1}^n I(U_i;Y_i|Y^{i-1}) \nonumber\\
       &\stackrel{(e)}\leq 
    \max_{{\{P(\Tilde{u}_i|\Tilde{u}_{i-1},y^{i-1}),x_i=f_i(\Tilde{u}_i,s_{i-1})
        \}}_{i=1}^n} \sum_{i=1}^n I(\Tilde{U}_i;Y_i|Y^{i-1}) \nonumber\\
     &\stackrel{(f)}= \max_{{\{P(\Tilde{u}_i|\tilde{u}_{i-1},y^{i-1}),x_i=f(s_{i-1},\Tilde{u}_i,i)\}}_{i=1}^n} \sum_{i=1}^n I(\tilde{U}_i;Y_i|Y^{i-1})  \nonumber\\
 &\stackrel{(g)}\leq \max_{{\{P(\dbtilde{u}_i|\dbtilde{u}_{i-1},y^{i-1}),x_i=f(\dbtilde{u}_i,s_{i-1})\}}_{i=1}^n}
 \sum_{i=1}^n I(\dbtilde{U}_i;Y_i|Y^{i-1})  \nonumber \\
 &\stackrel{(h)}\leq \max_{P(\dbtilde{u}^n||y^{n-1}),{\{x_i=f(\dbtilde{u}_i,s_{i-1})\}}_{i=1}^n}
 I(\dbtilde{U}^n \to Y^n),   \label{UB:last_step}
\end{align}
%for some  $f: \mathcal U \times \mathcal S \to \mathcal X$, where:\\
where $\epsilon_n\to 0$ as $n \to \infty$, and
\begin{enumerate}[label={(\alph*)}]
\item follows from Fano's inequality;
\item follows from defining $U_i\triangleq(M,Y^{i-1})$ for every $i \in [1:n]$. This definition satisfies $U_i=(U_{i-1},Y_{i-1})$ and the Markov chain $(Y_i,S_i)-(X_i,S_{i-1})-U_i$ due to the assumption that the channel is a FSC;
\item follows because the objective is determined by $\{P(u_i,y_i,s_i)\}_{i=1}^n$ due to the definition of $U_i$ and from the following lemma, whose proof appears in Appendix \ref{appendix:lemUpperBound}:
\begin{lemma}
\label{lem:UpperBound}
For any $k$, $P(u_k,y_k,s_k)$ is determined by $\{P(u_i|u_{i-1},y^{i-1})P(x_i|u_i,s_{i-1})\}_{i=1}^k$;
\end{lemma}
\item follows from the Functional Representation Lemma \cite[p.~626]{el2011network}, i.e., for every $i \in [1:n]$ there exists a RV $V_i$, such that $X_i$ can be represented as a function of $(U_i,S_{i-1},V_i)$, where $V_i$ is of cardinality $|\mathcal{V}_i|\leq |\cU_i \times \cS|(|\cX|-1)+1$, such that $V_i$ is independent of $(U_i,S_{i-1})$, and the Markov chain $(Y_i,S_i)-(U_i,S_{i-1},X_i)-V_i$ holds (then $(Y_i,S_i)-(X_i,S_{i-1})-(V_i,U_i)$ holds as well),
and from the following lemma:
\begin{lemma}
\label{lem:vtilde} 
For any $k$, $P(u_k,y_k,s_k)$ is determined by $\{P(u_i|u_{i-1},y^{i-1})P(v_i)x_i(v_i,u_i,s_{i-1})\}_{i=1}^k$;
\end{lemma}
The proof of Lemma \ref{lem:vtilde} is similar to that of Lemma \ref{lem:UpperBound} and therefore it is omitted;
\item follows from defining $\Tilde{U_i}\triangleq(U_i,V_i)$, from the fact that \\$P(\Tilde{u}_i|\Tilde{u}_{i-1},y^{i-1})=P(v_i|u_{i-1},v_{i-1},y^{i-1})P(u_i|v_i,u_{i-1},v_{i-1},y^{i-1})$. $P(v_i)$ and $P(u_i|v_i,u_{i-1},v_{i-1},y^{i-1})$ are sub-domains of $P(v_i|u_{i-1},v_{i-1},y^{i-1})$ and $P(u_i|u_{i-1},y^{i-1})$, respectively; the mutual information increases since conditioning reduces entropy;
\item follows since there exists an invariant function $f(\Tilde{u},s,i)=f_i(\Tilde{u}_i,s_{i-1})$.
\item follows from defining $\dbtilde{U}_i=(\Tilde{U}_i,T=i)$, where $T$ represents the time index.
\item follows since conditioning reduces entropy;
\end{enumerate}
Finally, we divide \eqref{UB:last_step} by $n$, rename $\dbtilde{U}$ by $U$, and obtain
\begin{align}
\label{eq:UBepsilonN}
R-\epsilon_n&\le \frac{1}{n} \max_{P(u^n||y^{n-1}),{\{x_i=f(u_i,s_{i-1})\}}_{i=1}^n}
 I(U^n \to Y^n),
\end{align}
which completes the proof of converse by tending $n\to \infty$.
\end{proof}

\subsection{Proof of Theorem \ref{theorem:anyN_LB_UB}}
\label{subsec:computableN}
First, we show the LHS of \eqref{eq:corLBUB}, then we show the RHS of it.

\textit{Proof of the LHS of \eqref{eq:corLBUB}:} Inducing a new FSC $P_{Y,S^+|U,S}$ with input $U$ without SI as was explained in the proof of Ineq. \eqref{eq:Ach1} and using the feedback capacity expression for such a setting as given in \eqref{eq:LB_anyFSC} imply that 
  \begin{align}
      C_{\text{fb-csi}}\ge \lim_{N\to \infty} \underline{C}_N=\sup_{N} \left[ \underline{C}_N-\frac{\log{\cS}}{N}\right],
      \label{eq:corLHS_LBUB}
  \end{align}
  where the equality above follows from the super-additive property of the sequence $N(\underline{C}_N-\frac{\log{\cS}}{N})$, which follows directly from the proof of 
  \cite[Th. 4]{PermuterWeissmanGoldsmith09} when replacing $X_i$ with $U_i$.
  
\textit{Proof of the RHS of \eqref{eq:corLBUB}:}
Combining \eqref{eq:UBepsilonN} and Lemma \ref{lemma:LemmaConnected} implies that 
\begin{align}
      C_{\text{fb-csi}}&\le \frac{1}{N} \max_{P(u^N||y^{N-1})} I(U^N \to Y^N|S_0) \nn\\
        &\le \lim_{N\to \infty} \overline{C}_N \nn\\
        &=\inf_{N} \left[ \underline{C}_N+\frac{\log{\cS}}{N}\right],
        \label{eq:corRHS_LBUB}
\end{align}
where the equality follows from the sub-additive property of the sequence $N(\overline{C}_N+\frac{\log{\cS}}{N})$, which follows directly from the proof of \cite[Th. 16]{PermuterWeissmanGoldsmith09} (by replacing $X_i$ with $U_i$).

Finally, \eqref{eq:corLBUB} is a direct consequence of \eqref{eq:corLHS_LBUB} and \eqref{eq:corRHS_LBUB}. 
\qed
%It is important to note that we do not have a general cardinality bound on the auxiliary RVs for the upper bounds in \eqref{eq:upperBoundLim} and \eqref{eq:upperBoundN}. For special cases $|\cU|$ is finite, and we suspect that for other special cases, such as the EH model, $|\cU|$ is unbounded.

\section{Conclusions}
\label{sec:conclusions}
The feedback capacity of connected FSCs with SI available causally at the encoder was derived. It is expressed as two equivalent multi-letter expressions which consist of a sequence of auxiliary RVs with memory. The first expression consists of DI and have a finite cardinality bound for the set of the auxiliary RVs, but the expression is complicated. The second expression does not have a finite cardinality bound, yet the expression is more simple since the auxiliary RVs constitute a first-order Markov process given the past outputs. Although both of the capacity expressions are multi-letter, we utilized them to provide computable lower and upper bounds on the feedback capacity. 
First, sequences of lower and upper bounds were given, i.e., for any integer $N$, finite-letter, computable bounds are obtained. Furthermore, by fixing a finite cardinality $|\cU|$ in the second multi-letter capacity expression, which renders it an achievable rate, two methods were given to compute lower bounds. The first method was a DP formulation of the achievable rate expression, and the second was a single-letter $Q$-graph lower bound. These methods were demonstrated on several examples and were shown to be useful in deriving achievable rates both analytically and numerically, and they were shown to be tight in some of the examples.

\appendices
\section{Proof of Lemma \ref{lem:DP_formulation}}
\label{appendix:lem_DP_proof}
In this proof, we show three parts regarding the state evolution, disturbance and reward of the DP.
\subsubsection{State evolution}
We shall prove that given a policy $\pi=(\mu_1,\mu_2,\dots)$, the new DP state is a time-invariant function of the current DP state, action and disturbance, i.e., there exists a function $F$ such that $z_i=F(z_{i-1},a_i,w_i)$, and in our case
$\beta_i=F(\beta_{i-1},a_i,y_i)$. For any $u_i,s_i \in \cU \times \cS$,
\begin{align}
\beta_i(u_i,s_i) &=P(u_i,s_i|y^i) \nonumber \\
&=\sum_{u_{i-1},s_{i-1}}P(u_i,s_i,u_{i-1},s_{i-1}|y^t) \nonumber \\ 
&=\frac{\sum\limits_{u_{i-1},s_{i-1}}P(u_i,s_i,u_{i-1},s_{i-1},y_i|y^{i-1})}{P(y_i|y^{i-1})} \nonumber\\
&=\frac{\sum\limits_{u_{i-1},s_{i-1}}P(u_i,s_i,u_{i-1},s_{i-1},y_i|y^{i-1})}{\sum\limits_{u'_i,u_{i-1},s_{i-1}}P(u'_i,u_{i-1},s_{i-1},y_i|y^{i-1})} \nonumber\\
&\stackrel{(a)}=\frac{\sum_{u_{i-1},s_{i-1}}P(u_{i-1},s_{i-1}|y^{i-1})P(u_i|u_{i-1},y^{i-1})P_{Y,S^+|X,S}(y_i,s_i|f(u_i,s_{i-1}),s_{i-1})}{\sum_{u_{i-1},u'_i,s_{i-1}} P(u_{i-1},s_{i-1}|y^{i-1})P(u'_i|u_{i-1},y^{i-1})P_{Y|X,S}(y_i|f(u'_i,s_{i-1}),s_{i-1})} \label{eq:BCJR} \\
&=\frac{\sum_{u_{i-1},s_{i-1}}\beta_{i-1}(u_{i-1},s_{i-1})a_i(u_i,u_{i-1},y^{i-1})P_{Y,S^+|X,S}(y_i,s_i|f(u_i,s_{i-1}),s_{i-1})}{\sum_{u_{i-1},u'_i,s_{i-1}}\beta_{i-1}(u_{i-1},s_{i-1})a_i(u_i',u_{i-1},y^{i-1})P_{Y|X,S}(y_i|f(u'_i,s_{i-1}),s_{i-1})} \label{eq:evolution}\text{,}
\end{align}
where (a) follows from the channel model and the Markov chain $U_i-(U_{i-1},Y^{i-1})-S_{i-1}$ implied from the joint distribution in \eqref{eq:joint_distTrunc}. Hence, there exists such a function $F$, i.e., $\beta_i=F(\beta_{i-1},a_i,y_i)$. \qed

\subsubsection{Disturbance}
We need to show that the disturbance satisfies the Markov chain $y_i-(\beta_{i-1},a_i)-(\beta^{i-2},a^{i-1},y^{i-1})$. 
Consider:
\begin{align}
P(y_i|\beta^{i-1},y^{i-1},a^i) &=\sum_{u_{i-1},s_{i-1},u_i} P(y_i,u_i,u_{i-1},s_{i-1}|\beta^{i-1},y^{i-1},a^i) \nonumber \\
&\stackrel{(a)}= \sum_{u_{i-1},s_{i-1},u_i} P(u_{i-1},s_{i-1}|\beta_{i-1},a_i)P(u_i|u_{i-1},s_{i-1},\beta_{i-1},a_i)P(y_i|u_i,u_{i-1},s_{i-1},\beta_{i-1},a_i) \nonumber \\
&= \sum_{u_{i-1},s_{i-1},u_i} P(y_i,u_i,u_{i-1},s_{i-1}|\beta_{i-1},a_i)\nonumber \\
&= P(y_i|\beta_{i-1},a_i)\nonumber {,}
\end{align}
where (a) follows due to $x_i=f(u_i,s_{i-1})$ and the channel model. \qed
\subsubsection{Reward}
Finally, we show that the reward is a time-invariant function of the current DP state and action.
Note that the reward depends only on the joint distribution $P(u_i,u_{i-1},y_i|y^{i-1})$. Consider:
\begin{align}
P(u_i,u_{i-1},y_i|y^{i-1})&=\sum_{s_{i-1}}P(u_i,u_{i-1},y_i,s_{i-1}|y^{i-1}) \nonumber \\
&\stackrel{(a)}=\sum_{s_{i-1}}P(u_{i-1},s_{i-1}|y^{i-1})P(u_i|u_{i-1},y^{i-1})P(y_i|f(u_i,s_{i-1}),s_{i-1}) \nonumber \\
    &= \sum_{s_{i-1}} \beta_{i-1}(u_{i-1},s_{i-1}) a_i(u_i,u_{i-1},y^{i-1})P(y_i|f(u_i,s_{i-1}),s_{i-1}) \label{eq:dp_reward} \text{,}
\end{align}
where (a) follows from the Markov chain $U_i-(U_{i-1},Y^{i-1})-S_{i-1}$ and the channel model. Therefore, $g(\beta_{i-1},a_i)=I(U_i,U_{i-1};Y_i|\beta_{i-1},a_i)\text{.}$ \qed
%%%%%%%%%%%%%%%%%%%%%%%%%%%%%%%%%%%%%%%%%%%%%%%%%%%%%%%%%%%%%%%%%
%%%%%
%%%%%%%%%%
%%%%%
%\begin{lemma}\label{lem:PolicyTrapdoorBEC}
%\eqref{eq:cap1} is the capacity of the Trapdoor channel with $p=0.5$ and the input-constrained binary erasure -- with %$\left|\mathcal{U}\right|=2$ and $x= u\oplus s$.
%\end{lemma}
%\textit{Note:} For validation, we let the VIA of the MDP run for two unifilar problems: the Trapdoor channel and the input-constrained BEC. $\left|\mathcal{U}\right|=2$ was sufficient to achieve the capacity, with $x= u\oplus s$. 

\section{Proof of Lemma \ref{lem:LB_supliminf}}
\label{appendix:LB_supliminf}
%\begin{proof}[Proof of Lemma \ref{lem:LB_supliminf}]
The following lemma is technical and will help to establish the proof of Lemma \ref{lem:LB_supliminf}.
\begin{lemma}
\label{lem:limInfHaim}
(Analogue to \cite[Lemma 4]{Permuter06_trapdoor_submit}) 
For any FSC with SI known causally at the encoder, the following equality holds:
\begin{align}
     \lim_{N\to \infty} \frac{1}{N} \max_{P(u^N||y^{N-1},s_0)} \min_{s_0} I(U^N \to Y^N|s_0)=
     \sup_{\{P(u_i|u^{i-1},y^{i-1},s_0)\}_{i\ge 1}} \liminf_{N\to\infty}   \frac{1}{N} \min_{s_0} I(U^N \to Y^N|s_0). \label{eq:limInfCap}
\end{align}
\end{lemma}
The proof of Lemma \ref{lem:limInfHaim} follows directly from the proof of
\cite[Lemma 4]{Permuter06_trapdoor_submit}, which relies on the super-additive property of the sequence $N(\underline{C}_N-\frac{\log{\cS}}{N})$. In our case, the super-additive property follows straightforwardly from the proof of \cite[Th. 4]{PermuterWeissmanGoldsmith09} by replacing $X_i$ with $U_i$.

\textit{Proof of Lemma \ref{lem:LB_supliminf}:} 
The proof consists of two parts. In the first part we show that 
\begin{equation}
    C_{\text{fb-csi}} = \sup_{{\{P(u_i|u^{i-1},y^{i-1})\}}_{i\geq1}} \liminf_{N\to\infty} \frac{1}{N} I(U^N \to Y^N), \label{eq:supliminfDirected}
\end{equation}
where $|\cU|=|\cX|^{|\cS|}$. In the second part, we use a transformation in order to render \eqref{eq:supliminfDirected} to \eqref{eq:supliminfcapCSI} for completing the proof. 

In the capacity proof of Theorem \ref{theorem:capDI} (Eqs. \eqref{eq:Ach1}-\eqref{eq:Ach3} in the achievability proof and Ineq. \eqref{eq:UBepsilonN} in the converse proof), it is shown that for the case of a connected FSC with SI known causally at the encoder, the LHS of \eqref{eq:limInfCap} characterizes the capacity and is equal to $\lim_{N \to \infty} \frac{1}{N}  \max_{P(u^N||y^{N-1})} I(U^N \to Y^N)$. By following the same arguments of \eqref{eq:Ach1}-\eqref{eq:Ach3} and \eqref{eq:UBepsilonN}, an equality between the RHS of \eqref{eq:limInfCap} and the RHS of \eqref{eq:supliminfDirected} also holds. Therefore, using Lemma \ref{lem:limInfHaim} we obtain \eqref{eq:supliminfDirected}.

It is possible to render \eqref{eq:supliminfDirected} to \eqref{eq:supliminfcapCSI}, i.e., with $|\cU_i|= |\cX|^{{|\cS|}^i}$ instead of $|\cU|= |\cX|^{|\cS|}$ and $(U_i,U_{i-1})$ shown in the expression instead of $U^i$, simply by the transformation given after Theorem \ref{theorem:capTrunc} in Section~\ref{sec:main_results}. 
\qed

\section{The SS Does Not Increase The Capacity of the I.I.D. $ZS$-Channel}
\label{appendix:CSInotIncreaseZScapacity}
\begin{theorem}
\label{theorem:noisyPOSTcap}
For the i.i.d. $ZS$-channel, SI available causally at the encoder does not increase the capacity.
\end{theorem}

\begin{proof}[Proof of Theorem \ref{theorem:noisyPOSTcap}]
\label{appendix:ProofOfCardinalityNoisyPOST}
%The derivation of the capacity formula in \eqref{eq:noistPOSTcaphalf} follows by straightforward parameterization of the channel input and computation of the stationary distribution.
Recall the SS capacity formula:  $C_{\text{CSI-E}}=\max_{p(u),x(u,s)}I(U;Y)$. Consider $|\cU|=|\cX|^{|\cS|}=4$ with all possible strategies as detailed in Table \ref{table:strategies}. 
\begin{table}[t]
\caption{All the strategies of binary input and binary state alphabets, $\cS=\cX=\{0,1\}$.} \centering
\label{table:strategies}
\begin{tabular}[b]{||c|c|c|c||}
\hline \hline
$x(u,s)$ & $s=0$ & $s=1$ \\
\hline \hline
$u_0$ & $0$ & $0$ \\
\hline
$u_1$ & $0$ & $1$ \\
\hline
$u_2$ & $1$ & $0$ \\
\hline
$u_3$ & $1$ & $1$ \\
\hline \hline
\end{tabular}
\end{table}
Assume that $I_{P_1}(U;Y)$ is the capacity of this channel with SI available causally at the encoder, induced by some input distribution $P_1(u)$ with the corresponding joint distribution
\begin{align}
    &P_1(u,x,y) =\sum_{s} P(s)P_1(u)\mathbbm{1}\{x=f(u,s)\}P(y|x,s), \nn
\end{align}
We construct an input distribution $P_2(x)$ with the corresponding conditional mutual information satisfying $I_{P_2}(X;Y)=I_{P_1}(U;Y)$ induced by the joint distribution
\begin{align}
    P_2(x,y)&=P_2(x) \sum_s P(s) P(y|x,s), 
\end{align}
that is, $X$ is independent of $S$. Clearly, $I_{P_1}(U;Y)\ge I_{P_2}(X;Y)$; thus, our goal is to show that 
$I_{P_1}(U;Y)\le I_{P_2}(X;Y)$. In the construction of $P_2(x)$, we only demand that it satisfies
\begin{align}
    P_2(x)=P_1(x) \quad \forall x\in \cX, \label{eq:preserveP12}
\end{align}
where $P_1(x)$ is the input distribution induced by $P_1(u)$, and given by 
\begin{align}
    P_1(x)&= \sum_{u,s} P_1(u,s,x)=\sum_{s} P(s) \sum_u P_1(u) \mathbbm{1}\{x=f(u,s)\}. \nn
\end{align}
Hence, for our channel we obtain
% \begin{align}
%     P_2(x|y')= P_1(x|y')=\sum_{u,s} P_1(u,s,x|y')=\sum_{u,s} P_1(u|y')P(s|y')\mathbbm{1}\{x=f(u,s)\},
% \end{align}
\begin{align}
    P_2(X=1)&\triangleq \textstyle \frac{1}{2}[P_1(u_1)+P_1(u_2)]+P_1(u_3) , \label{eq:P2_a}    
\end{align}
From the construction in \eqref{eq:preserveP12}, it follows that the output distributions are also equal, i.e., $P_2(y)=P_1(y)$:

\begin{align}
    P_2(Y=1)&=\sum_{s,x} P(s) P_2(x) P(y|x,s)=\textstyle \frac{1}{4}P_2(x=0)+\frac{3}{4}P_2(x=1)=\frac{1}{2}P_2(x=1)+\frac{1}{4}\nn\\
    &=\textstyle \frac{1}{4}P_1(u_1)+\frac{1}{4}P_1(u_2)+\frac{1}{2}P_1(u_3)+\frac{1}{4} \nn\\
    &= \textstyle \frac{1}{4}P_1(u_0)+\frac{1}{2}P_1(u_1)+\frac{1}{2}P_1(u_2)+\frac{3}{4}P_1(u_3) \nn\\
    &=\sum_{s,u,x} P(s) P_1(u) \mathbbm{1}\{x=f(u,s)\} P(y|x,s) \nn\\
    &=P_1(Y=1)
    %\label{eq:PresOutputs}
\end{align}
Consequently, $H_{P_2}(Y)=H_{P_1}(Y)$ hold; thus,
\begin{align}
    I_{P_1}(U;Y)&=H_{P_2}(Y)-H_{P_1}(Y|U) \nn\\
    & \stackrel{(a)} \le H_{P_2}(Y)-H_{P_2 }(Y|X) \nn\\
    &= I_{P_2}(X;Y) , \nn %\label{eq:lemmaCSI2steps}
\end{align}
where (a) follows from defining 
%Step $(ii)$ in \eqref{eq:lemmaCSI2steps} follows from
$q \triangleq  H_{P_1}(Y|U)-H_{P_2}(Y|X)\ge 0$. 
We show that $q\ge 0$ by applying the channel model as follows:
\begin{align}
    H_{P_1}(Y|U)&=\textstyle [P_1(u_0)+P_1(u_3)]H(\frac{1}{4})+P_1(u_1)+P_1(u_2)\nn\\
    &=\textstyle[1-P_1(u_1)-P_1(u_2)]H(\frac{1}{4})+P_1(u_1)+P_1(u_2) \nn\\
    H_{P_2}(Y|X)&=\textstyle P_2(x=0) H(\frac{1}{4})+P_2(x=1) H(\frac{3}{4})=H(\frac{1}{4}).
\end{align}
%where (a) and (d) follow from the construction of $P_2(x|y')$
%in \eqref{eq:P2_a}-\eqref{eq:P2_b}; and (b) and (c) follow from substituting $P_1(u_3)$ and $P_2(X=1)$, respectively, with their complementary distribution to $1$.
%, which also implies that $H_{P_2}(Y|U)=H_{P_1}(Y|U)$.
Hence, we deduce that
\begin{align}
    q&=\textstyle [P_1(u_1)+P_1(u_2)](1-H(\frac{1}{4}))     \ge 0.
\end{align}
To conclude, $I_{P_2}(X;Y)=I_{P_1}(U;Y)$, which implies that causal SI available at the encoder does not increase the capacity of the i.i.d. $ZS$-channel.
\end{proof}

\section{Proof of Corollary \ref{corollary:trapdoor}}
\label{appendix:ProofCorTrapdoor}
\begin{proof}
Define $b_1 \triangleq\sqrt{5}-2$, $b_2 \triangleq\frac{3-\sqrt{5}}{2}$, $b_3 \triangleq\frac{\sqrt{5}-1}{2}$, $b_4 \triangleq 3-\sqrt{5}$.
The proof follows by application of Theorem \ref{theorem:qgraph_LB}. Constructing the corresponding $(S,U,Q)$-graph and straightforward calculation of Eq. \eqref{eq:suq_transition} give the transition matrix $P(s^+,u^+,q^+|s,u,q)$, which has a unique stationary distribution $\pi(s,u,q)$ with the marginal distribution of the $Q$-graph: 
\begin{align}
    \mu\triangleq[\pi(Q=1), \pi(Q=2), \pi(Q=3), \pi(Q=4)]
    =\bigg[\frac{b_2}{b_4+2},\frac{1}{b_4+2},\frac{b_2}{b_4+2},\frac{1}{b_4+2} \bigg] \nn
\end{align}
and the conditional distribution:
%Further, one can calculate that for $Q=1$:
\begin{align}
  &[\pi_{U,S|Q}(0,0|1),\pi_{U,S|Q}(1,0|1),\pi_{U,S|Q}(0,1|1),\pi_{U,S|Q}(1,1|1)]= [0,b_2,b_3,0]\nn\\
  &[\pi_{U,S|Q}(0,0|2),\pi_{U,S|Q}(1,0|2),\pi_{U,S|Q}(0,1|2),\pi_{U,S|Q}(1,1|2)]= [0,b_1,b_4,0]\nn\\
  &[\pi_{U,S|Q}(0,0|3),\pi_{U,S|Q}(1,0|3),\pi_{U,S|Q}(0,1|3),\pi_{U,S|Q}(1,1|3)]= [b_3,0,0,b_2]\nn\\
  &[\pi_{U,S|Q}(0,0|4),\pi_{U,S|Q}(1,0|4),\pi_{U,S|Q}(0,1|4),\pi_{U,S|Q}(1,1|4)]= [b_4,0,0,b_1]\nn,
\end{align}
and the BCJR-invariant property for each node can be verified. The reward of each node, i.e., $I(U^+,U;Y|Q=i)$ for $i=1,2,...,|\cQ|$, is induced from
\begin{align}
P(u^+,u,y|q)&= \sum_{s_{i-1}} \pi(u,s|q) P(u^+|u,q)P(y|f(u^+,s),s) \text{,}
\end{align}
yielding:
%\begin{align}
 $r\triangleq[I(U^+,U;Y|Q=1), I(U^+,U;Y|Q=2), I(U^+,U;Y|Q=3), I(U^+,U;Y|Q=4)]$, %\nn
%\end{align}
where $I(U^+,U;Y|Q=i)=1-b_2 b_3$ for $i=1,3$, and $I(U^+,U;Y|Q=j)=H\big(b_2(b_3+1)\big)-b_2 b_4$ for $j=2,4$. Finally, $\mu \cdot r^T= \log \phi$ is a lower bound on the capacity.
\end{proof}
\section{Proof of Corollary \ref{corollary:ising}}
\label{appendix:ProofCorIsing}
\begin{proof}
Constructing the corresponding $(S,U,Q)$-graph and straightforward calculation of Eq. \eqref{eq:suq_transition} give the transition matrix $P(s^+,u^+,q^+|s,u,q)$, which has a unique stationary distribution $\pi(s,u,q)$ with the marginal distribution of the $Q$-graph: 
\begin{align}
    \mu\triangleq[\pi(Q=1), \pi(Q=2), \pi(Q=3), \pi(Q=4)]
    =\bigg[\frac{a+1}{2(a+3)},\frac{1}{a+3},\frac{a+1}{2(a+3)},\frac{1}{a+3} \bigg] \nn
\end{align}
and the conditional distribution:
%Further, one can calculate that for $Q=1$:
\begin{align}
  &[\pi_{U,S|Q}(0,0|1),\pi_{U,S|Q}(1,0|1),\pi_{U,S|Q}(0,1|1),\pi_{U,S|Q}(1,1|1)]= \bigg[\frac{\bar{a}}{a+1},0,0,\frac{2a}{a+1}\bigg]\nn\\
  &[\pi_{U,S|Q}(0,0|2),\pi_{U,S|Q}(1,0|2),\pi_{U,S|Q}(0,1|2),\pi_{U,S|Q}(1,1|2)]= [1,0,0,0]\nn\\
  &[\pi_{U,S|Q}(0,0|3),\pi_{U,S|Q}(1,0|3),\pi_{U,S|Q}(0,1|3),\pi_{U,S|Q}(1,1|3)]= \bigg[\frac{2a}{a+1},0,0,\frac{\bar{a}}{a+1}\bigg]\nn\\
  &[\pi_{U,S|Q}(0,0|4),\pi_{U,S|Q}(1,0|4),\pi_{U,S|Q}(0,1|4),\pi_{U,S|Q}(1,1|4)]= [0,0,0,1]\nn,
\end{align}
and the BCJR-invariant property for each node can be verified. The reward of each node, i.e., $I(U^+,U;Y|Q=i)$ for $i=1,2,\dots,|\cQ|$, is 
% induced from
% \begin{align}
% P(u^+,u,y|q)&= \sum_{s_{i-1}} \pi(u,s|q) P(u^+|u,q)P(y|f(u^+,s),s) \text{,}
% \end{align}
% yielding:
\begin{align}
r\triangleq&[I(U^+,U;Y|Q=1), I(U^+,U;Y|Q=2), I(U^+,U;Y|Q=3), I(U^+,U;Y|Q=4)], \nn
\end{align}
where $I(U^+,U;Y|Q=i)=H\big(\frac{2a}{a+1}\big)$ for $i=1,3$, and $I(U^+,U;Y|Q=j)=H\big(\frac{a+1}{2}\big)+a-1$ for $j=2,4$. Finally, $\mu \cdot r^T= \frac{2 H(a)}{a+3}$ is a lower bound on the capacity.
\end{proof}

\section{Proof of Corollary \ref{corollary:BEC}}
\label{appendix:ProofCorBEC}
\begin{proof}
Throughout the proof, we use the notations $H_2(\alpha),H_3(\alpha_1,\alpha_2)$ for the distinction between the binary and ternary entropies (since $|\cY|=3$), respectively, where $\alpha,\alpha_1,\alpha_2 \in [0,1]$ and $\alpha_1+\alpha_2\le 1$.
The proof follows by application of Theorem \ref{theorem:qgraph_LB} with $\left|\mathcal{U}\right|=2$, function
(which respects the input constraint) the $Q$-graph from Fig. \ref{fig:Qgraph_ising} and the auxiliary conditional distribution
Constructing the corresponding $(S,U,Q)$-graph and straightforward calculation of Eq. \eqref{eq:suq_transition} give the transition matrix $P(s^+,u^+,q^+|s,u,q)$, which has a unique stationary distribution $\pi(s,u,q)$ with the marginal distribution of the $Q$-graph: 
\begin{align}
    \mu\triangleq[\pi(Q=1), \pi(Q=2), \pi(Q=3)]
    =\bigg[\frac{p\bar{\epsilon}}{1+p\bar{\epsilon}},\frac{\bar{\epsilon}}{1+p\bar{\epsilon}},\frac{\epsilon}{1+p\bar{\epsilon}} \bigg] \nn
\end{align}
and the conditional distribution:
%Further, one can calculate that for $Q=1$:
\begin{align}
  &[\pi_{U,S|Q}(0,0|1),\pi_{U,S|Q}(1,0|1),\pi_{U,S|Q}(0,1|1),\pi_{U,S|Q}(1,1|1)]= [0,0,0,1]\nn\\
  &[\pi_{U,S|Q}(0,0|2),\pi_{U,S|Q}(1,0|2),\pi_{U,S|Q}(0,1|2),\pi_{U,S|Q}(1,1|2)]= [1,0,0,0]\nn\\
  &[\pi_{U,S|Q}(0,0|3),\pi_{U,S|Q}(1,0|3),\pi_{U,S|Q}(0,1|3),\pi_{U,S|Q}(1,1|3)]= [\bar{p},0,0,p]\nn
\end{align}
where $\pi_{U,S|Q}(u,s|q)=\pi_{U,S|Q}(U=u,S=s|Q=q)$, and the BCJR-invariant property for each node can be verified. The reward of each node, i.e., $I(U^+,U;Y|Q=i)$ for $i=1,2,\dots,|\cQ|$, is 
\begin{align}
  r\triangleq&[I(U^+,U;Y|Q=1), I(U^+,U;Y|Q=2), I(U^+,U;Y|Q=3)], \nn
\end{align}
where $I(U^+,U;Y|Q=1)=0$, 
$I(U^+,U;Y|Q=i)=H_3(P(\bar{\epsilon}),\epsilon)-H_2(\epsilon)$ for $i=2,3$. 
Finally, $\mu \cdot r^T=\frac{H(p)}{\frac{1}{\bar{\epsilon}}+p}$ is a lower bound on the capacity for any $p$.
\end{proof}

\section{Proof of Lemma \ref{lemma:LemmaConnected}}
\label{appendix:proof_LemmaConnected}
The proof follows exactly in the same steps of \cite[Lemma~2]{Permuter06_trapdoor_submit}, yet with the following modifications. Replace each occurrence of $X_i$ with $U_i$, and construct $P_2(u^N||y^{N-1},s'_0)$ as follows. First, construct $P_2(u^N||y^{N-1},s'_0)$ due to an input distribution of the form $\{P(x_i|s_{i-1})\}_{i=1}^T$ with positive probability of reaching $s_0$ in $T$ time epochs, i.e., $\sum_{i=1}^T P(S_i=s_0|S_0=\tilde{s})>0, \tilde{s}\in \cS$, denoted by $\tilde{P}$ \footnote{Such an input distribution $\tilde{P}$ exists because the FSC is assumed to be strongly connected.}, repeatedly, until the time epoch that the channel first reaches $s_0$, denoted by time $L$. In particular, until time $L$, the encoder uses only the $|\cX|$ strategies $u$ that map all states to a specific input $x$. After each time epoch within a time-window $T$, the encoder observes the next reached state $s_{i-1}$ and constructs $P(u_i)=P(u_i|s_{i-1})$ exactly according to the next time epoch in $\tilde{P}$. This is operatively possible because of the SI known causally at the encoder at each time epoch. After time $L$, construct $P_2$ exactly as $P_1$ would (had time started then), i.e., for $i>L$:
\begin{align}
    P_2(u_i|u^{i-1},y^{i-1},s_0)=P_1(u_{i-L}|u^{i-1-L},y^{i-1-L},s_0).
\end{align}
The proof is concluded by following \cite[Inequality (13)]{Permuter06_trapdoor_submit}.

\section{Proof of lemma \ref{lem:UpperBound}.}
\label{appendix:lemUpperBound}
We need to show that for a FSC with feedback and SI known causally at the encoder, for any $k$, the joint distribution $P(u_k,y_k,s_k)$, where $u_i\triangleq(m,y^{i-1}), i\in [1:k]$,
is determined by $\{P(u_i|u_{i-1},y^{i-1})P(x_i|u_i,s_{i-1})\}_{i=1}^k$. We prove it by induction.
For $n=1$ we have:
\begin{align}
    P(u_1,y_1,s_1)=\sum_{s_0,x_1} P(y_1,s_1|x_1,s_0)P(x_1|u_1,s_0)P(u_1)P(s_0), \nonumber 
\end{align}
because from the definition $U_i\triangleq(M,Y^{i-1})$ it follows that $U_1=M$ is independent of $S_0$, and $(Y_1,S_1)-(X_1,S_0)-U_1$ form a Markov chain due to the FSC definition. Assume $P(u_{k-1},y_{k-1},s_{k-1})$ is determined by $\{P(u_i|u_{i-1},y^{i-1}),P(x_i|u_i,s_{i-1})\}_{i=1}^{k-1}$, we have %(correct the sum eliko!):
\begin{align}
P(u_k,y_k,s_k) =\sum_{s_{k-1}}P(u_k,y_k,s_k|u_{k-1},y_{k-1},s_{k-1})P(u_{k-1},y_{k-1},s_{k-1}), \nonumber
\end{align}
and it is sufficient to show that $P(u_k,y_k,s_k|u_{k-1},y_{k-1},s_{k-1})$ depends on $\{P(u_i|u_{i-1},y_{i-1}),P(x_i|u_i,s_{i-1})\}_{i=1}^k$:
\begin{align}
    P(u_k,y_k,s_k|u_{k-1},y_{k-1},s_{k-1})&=\sum_{x_k}P(u_k,y_k,s_k,x_k|u_{k-1},y_{k-1},s_{k-1})\nonumber\\
    &\stackrel{(a)}=\sum_{x_k}P(u_k|u_{k-1},y_{k-1})P(x_k|u_k,s_{k-1})P(y_k,s_k|x_k,s_{k-1},u_k)\nonumber\\
    &\stackrel{(b)}=\sum_{x_k}P(u_k|u_{k-1},y_{k-1})P(x_k|u_k,s_{k-1})P(y_k,s_k|x_k,s_{k-1})\nonumber\\ \nonumber
\end{align}
where
\begin{enumerate}[label={(\alph*)}]
\item follows from $U_k=(U_{k-1},Y_{k-1})$ due to the definition of $U_k$.
\item follows from the Markov chain $(Y_k,S_k)-(X_k,S_{k-1})-U_k$ that holds due to the definition of $U_k$ and the fact that the channel is a FSC.
\end{enumerate}
\qed

\section{Proof of $|\cU|=|\cX|^{|\cS|}$} 
\label{appendix:cardinality}
\begin{proof}
Let $T=\{x_u (s):\cS \to \cX|u\in \{0,1,\dots,|\cX|^{|\cS|}-1\}\}$ be the set of all $|\cX|^{|\cS|}$ different functions. 
Assume, to the contrary, that 
\eqref{eq:capDI} is achieved with distribution $P_1(\tilde{u}^N,y^N)$ where at least one variable $U_i, i\in [1:N]$ has cardinality greater than $|\mathcal{T}|$ (i.e., $|\cU_i|>|\cX|^{|\cS|}$ for some $i$), and that it is greater than the objective induced by any distribution $P(u^N,y^N)$ with $|\cU_i|=|\mathcal{T}|=|\cX|^{|\cS|}, \forall i\in [1:N]$. Our goal is to construct such distribution $P_2(u^N,y^N)$ that achieves \eqref{eq:capDI} as well. The DI induced by $P_1(\tilde{u}^N,y^N)$ can be written as
\begin{align}
    I_{P_1}(U^N \to Y^N|s_0)&= \frac{1}{N}\sum_{i=1}^{N} I_{P_1}(U^i;Y_i|Y^{i-1},s_0) \nn\\
    &= \frac{1}{N}\sum_{i=1}^{N} H_{P_1}(Y_i|Y^{i-1})-H_{P_1}(Y_i|Y^{i-1},U^i), \label{eq:DIcondEnt}
\end{align}
where, using the definition of conditional entropy, the entropies in \eqref{eq:DIcondEnt} can explicitly be written as
\begin{align}
    &H_{P_1}(Y_i|Y^{i-1})= -\sum_{y^i} \sum_{\tilde{u}^i\in {\cU}^i} P_1(\tilde{u}^i,y^i) \log P_1(y_i|y^{i-1}) , \label{eq:H1} \\
    &H_{P_1}(Y_i|Y^{i-1},U^{i})=-\sum_{y^i} \sum_{\tilde{u}^i\in {\cU}^i}  P_1(\tilde{u}^i,y^i) \log P_1(y_i|\tilde{u}^{i},y^{i-1}).\label{eq:H2}
\end{align}
%where $\sum_{\tilde{u}^i}$ is used rather than $\sum_{\tilde{u}^i}$ since some variables $U_i$ can be of an unbounded cardinality.
We are going to show that a legitimate distribution $P_2(u^N,y^N)$ satisfies $H_{P_2}(Y_i|Y^{i-1})=H_{P_1}(Y_i|Y^{i-1})$ and $H_{P_2}(Y_i|Y^{i-1},U^{i})=H_{P_1}(Y_i|Y^{i-1},U^{i})$. We focus on \eqref{eq:H2} first, in which
\begin{align}
P_1(y_i|\tilde{u}^{i},y^{i-1})&\stackrel{(a)}=\sum_{s_{i-1}} P_1(s_{i-1}|\tilde{u}^{i},y^{i-1})P_1(y_i|\tilde{u}^{i},s_{i-1},y^{i-1}) \nn \\
&\stackrel{(b)}=\sum_{s_{i-1}} P_1(s_{i-1}|\tilde{u}^{i},y^{i-1})P(y_i|f(\tilde{u}_i,s_{i-1}),s_{i-1}) \nn \\
&\stackrel{(c)}=\sum_{s_{i-1}} P_1(s_{i-1}|\tilde{u}^{i-1},y^{i-1})P(y_i|f(\tilde{u}_i,s_{i-1}),s_{i-1}) \nn\\
&\stackrel{(d)}=\sum_{s_{i-1}}  
\frac{\sum_{s^{i-2}}\prod_{j=1}^{i-1} P(y_j,s_j|f(\tilde{u}_j,s_{j-1}),s_{j-1}) P(y_i|f(\tilde{u}_i,s_{i-1}),s_{i-1})}{\sum_{s^{i-2}}\big(\prod_{j=1}^{i-2} P(y_j,s_j|f(\tilde{u}_j,s_{j-1}),s_{j-1})\big) P(y_{i-1}|f(\tilde{u}_{i-1},s_{i-2}),s_{i-2})}, \label{eq:yGivenPast}
\end{align}
where
\begin{enumerate}[label={(\alph*)}]
\item follows follows from the law of total probability;
\item follows from the channel model;
\item follows from the Markov chain $U_i-(U^{i-1},Y^{i-1})-S_{i-1}$;
\item follows since $P_1(s_{i-1}|\tilde{u}^{i-1},y^{i-1})=\frac{P_1(s_{i-1},\tilde{u}^{i-1},y^{i-1})}{P_1(\tilde{u}^{i-1},y^{i-1})}$ where
\begin{align}
P_1(s_{i-1},\tilde{u}^{i-1},y^{i-1})&=( \prod_{j=1}^{i-1} P_1(\tilde{u}_j|\tilde{u}^{j-1},y^{j-1}) ) \sum_{s^{i-2}} \prod_{j=1}^{i-1} P(y_j,s_j|f(\tilde{u}_j,s_{j-1}),s_{j-1}), \nn\\
P_1(\tilde{u}^{i-1},y^{i-1})&=( \prod_{j=1}^{i-1} P_1(\tilde{u}_j|\tilde{u}^{j-1},y^{j-1}) ) \sum_{s^{i-2}} \prod_{j=1}^{i-2} P(y_j,s_j|f(\tilde{u}_j,s_{j-1}),s_{j-1}) \nn\\ &\quad \times P(y_{i-1}|f(\tilde{u}_{i-1},s_{i-2}),s_{i-2}). \nn
\end{align}
\end{enumerate}
That is, $P_1(y_i|\tilde{u}^{i},y^{i-1})$ does not depend on the actual symbols of $\tilde{u}^i$; rather, it depends on the strategies $(f(\tilde{u}_1,s_{0}),f(\tilde{u}_2,s_{1}),\dots,f(\tilde{u}_i,s_{i-1}))$ for all $s_0^{i-1} \in \cS^i$. Hence, it has a single, specific $u^{i}$ of all different $|\cX|^{{|\cS|}^i}$ vectors such that $P_1(y_i|\tilde{u}^{i},y^{i-1})=P_2(y_i|u^{i},y^{i-1})$, where the latter is defined similarly to \eqref{eq:yGivenPast} as
\begin{align}
P_2(y_i|u^{i},y^{i-1})&\stackrel{(a)}=\sum_{s_{i-1}}  
\frac{\sum_{s^{i-2}}\prod_{j=1}^{i-1} P(y_j,s_j|f(u_j,s_{j-1}),s_{j-1}) P(y_i|f(u_i,s_{i-1}),s_{i-1})}{\sum_{s^{i-2}}\big(\prod_{j=1}^{i-2} P(y_j,s_j|f(u_j,s_{j-1}),s_{j-1})\big) P(y_{i-1}|f(u_{i-1},s_{i-2}),s_{i-2})}. \label{eq:yGivenPast2}
\end{align}
% , where $|\cU|=|\cX|^{|\cS|}$
% which has a single
% x_u (s)
% of $\tilde{u}^i$, i.e., the vector $x^i=(f(\tilde{u}_1,s_{0}),f(\tilde{u}_2,s_{1}),\dots,f(\tilde{u}_i,s_{i-1}))$,
% for a fixed $y^i$ and for all $\tilde{u}^{i}$ with the same mappings $x^i$, one can replace $\log P(y_i|\tilde{u}^{i},y^{i-1})$ in \eqref{eq:H2} by $P(y_i|u^{i},y^{i-1})$ with a single, specific $u^{i}$ that has the same mappings, where $|\cU|=|\cX|^{|\cS|}$, such that \eqref{eq:H2} remains the same.
Here, we denote the $j$-th vector of strategies by $u^i_{(j)}$ (in this context, it is unimportant how we index the vectors, but it is important to note that $u^i_{(j)}$ and $u^i_{(k)}$ are two different vector for $j\ne k$). Further, we denote some $\tilde{u}^i$ with the same strategies as of a specific $u^i_{(j)}$ by $\tilde{u}_{(j)}^i$.
% $\tilde{u}_{k}^i$ and the $u^i$ that belongs to the same vector $x^i_{(j)}$ by
% $\tilde{u}_{(k)}^{i}$ and $u_{(k)}^{i}$, respectively. 
Consequently, by constructing $P_2(u^i_{(k)},y^i)\triangleq \sum_{\tilde{u}^i_{(k)}} P_1(\tilde{u}^i_{(k)},y^i)$
for a fixed $y^i$ and for all $k=1,2,\dots,|\cX|^{{|\cS|}^{i}}$, we obtain that any term in \eqref{eq:H2} satisfies
\begin{align}
     \sum_{\tilde{u}^i }  P_1(\tilde{u}^i,y^i) \log P_1(y_i|\tilde{u}^{i},y^{i-1})&=\sum_{k=1}^{|\cX|^{{|\cS|}^{i}}} \sum_{\tilde{u}_{(k)}^i }  P_1(\tilde{u}_{(k)}^i,y^i) \log P_1(y_i|\tilde{u}_{(k)}^{i},y^{i-1})\nn\\
     &=\sum_{k=1}^{|\cX|^{{|\cS|}^{i}}}\sum_{\tilde{u}_{(k)}^i }  P_1(\tilde{u}_{(k)}^i ,y^i) \log P_2(y_i|u_{(k)}^{i} ,y^{i-1}) \nn\\
    &=\sum_{k=1}^{|\cX|^{{|\cS|}^{i}}} \log P_2(y_i|u_{(k)}^{i},y^{i-1}) \sum_{\tilde{u}_{(k)}^i}  P_1(\tilde{u}_{(k)}^i ,y^i) \nn\\
    &=\sum_{k=1}^{|\cX|^{{|\cS|}^{i}}} P_2(u_{(k)}^i,y^i) \log P_2(y_i|u_{(k)}^{i} ,y^{i-1}) . 
\end{align}
Further, the construction of $P_2(u^i_{(k)},y^i)$ for all $k$ and all $y^i$ fully attains a constructed joint distribution $P_2(u^i,y^i)$. We repeat this construction for all $i \in [1,\dots,N]$, and then we uniquely get the conditional distribution $P_2(u_i|u^{i-1},y^{i-1})$ by the relation 
\begin{align}
    P_2(u^i,y^i)=P_2(u^{i-1},y^{i-1})P_2(u_i|u^{i-1},y^{i-1})P_2(y_i|u^i,y^{i-1}),
\end{align}
where $P_2(u^i,y^i),P(u^{i-1},y^{i-1})$ are known from the previously described construction, and $P_2(y_i|u^i,y^{i-1})$ is given by \eqref{eq:yGivenPast2}. As a result, we have 
a legitimate distribution $P_2(u^N,y^N)$ with $|\cU|=|\cX|^{|\cS|}$ that induces the same value of $H(Y_i|Y^{i-1},U^i)$ as induced by $P_1(\tilde{u}^N,y^N)$ in \eqref{eq:H2}, for all $i \in [1,\dots,N]$. 

It remains to prove that $P_2(u^N,y^N)$ induces 
$H_{P_2}(Y_i|Y^{i-1})=H_{P_1}(Y_i|Y^{i-1})$.
%In order to distinguish between the marginal of $y^N$ induced by $P(\tilde{u}^N,y^N)$ and the one induced by the constructed $P(u^N,y^N)$, we use the notation $P_1(y^N)$ for the former, and the notation $P_2(y^N)$ for the latter.
Now,
\begin{align}
    P_1(y_i|y^{i-1})&=\frac{P_1(y^{i})}{P_1(y^{i-1})} \nn\\
    &=\frac{\sum_{\tilde{u}^i\in \cU^i} P_1(\tilde{u}^i,y^i)}{\sum_{\tilde{u}^{i-1}\in \cU^{i-1}} P_1(\tilde{u}^{i-1},y^i)} \nn\\
    &=\frac{\sum_{k=1}^{|\cX|^{{|\cS|}^{i}}} \sum_{\tilde{u}^i_{(k)}} P_1(\tilde{u}^i_{(k)},y^i)}{\sum_{k=1}^{|\cX|^{{|\cS|}^{i-1}}} \sum_{\tilde{u}^{i-1}_{(k)}} P_1(\tilde{u}^{i-1}_{(k)},y^i)} \nn\\
    &=\frac{\sum_{k=1}^{|\cX|^{{|\cS|}^{i}}} P_2(u^i_{(k)},y^i)}{\sum_{k=1}^{|\cX|^{{|\cS|}^{i-1}}} P_2(u^{i-1}_{(k)},y^i)} \nn\\
    &=\frac{P_2(y^{i})}{P_2(y^{i-1})} \nn\\
    &=P_2(y_i|y^{i-1}).
\end{align}
Hence, for a fixed $y^i$, any term in \eqref{eq:H1} satisfies 
\begin{align}
    \sum_{\tilde{u}^i \in \cU^i} P_1(\tilde{u}^i,y^i) \log P_1(y_i|y^{i-1})&=\sum_{k=1}^{|\cX|^{{|\cS|}^{i}}} 
    \sum_{\tilde{u}^i_{(k)}} P_1(\tilde{u}^i_{(k)},y^i) \log P_1(y_i|y^{i-1}) \nn\\
    &=\sum_{k=1}^{|\cX|^{{|\cS|}^{i}}} 
    P_2(u^i_{(k)},y^i) \log P_1(y_i|y^{i-1}) \nn\\
    &=\sum_{k=1}^{|\cX|^{{|\cS|}^{i}}} 
    P_2(u^i_{(k)},y^i) \log P_2(y_i|y^{i-1}), 
\end{align}
and it holds for all $i \in [1:N]$. Finally, $P_2(u^N,y^N)$ achieves \eqref{eq:capDI} as well, and the supposition is false.
\end{proof}


\bibliographystyle{IEEEtran}
%\bibliography{IEEEabrv,ref}
\bibliography{FSC_submit}

\end{document}